\documentclass[envcountsame,envcountsect]{llncs}
\usepackage[english]{babel}
\usepackage[utf8x]{inputenc}
\usepackage[T1]{fontenc}
\usepackage[a4paper,top=2.7cm,bottom=2.5cm,left=3cm,right=3cm,marginparwidth=1.75cm]{geometry}

\usepackage{enumitem}
\usepackage[skins]{tcolorbox}
\usepackage{amsmath}
\usepackage{mathtools}
\usepackage{graphicx}
\usepackage{amssymb}
\usepackage{amsfonts}
\usepackage[colorinlistoftodos]{todonotes}
\usepackage{breakcites}
\usepackage{upgreek}
\usepackage{bbold}
\usepackage{dsfont}
\usepackage{stmaryrd}
\usepackage{wasysym}
\usepackage{framed}
\usepackage{empheq}
\usepackage{xcolor}
\usepackage{theorem}
\usepackage{algpseudocode}
\usepackage{refcount}
\usepackage{qcircuit}

\DeclareMathOperator*{\E}{\mathbb{E}}

\usepackage{float}
\usepackage{xspace}

\usepackage[colorlinks=true, allcolors=blue,breaklinks=true]{hyperref}

\usepackage[normalem]{ulem}

\usepackage{tikz}

\usetikzlibrary{decorations.pathmorphing}
\tikzset{snake it/.style={decorate, decoration=snake}}

\newcommand{\bra}[1]{\langle #1|}
\newcommand{\ket}[1]{|#1\rangle}
\newcommand{\Ket}[1]{\big|#1\big\rangle}

\newcommand{\braket}[2]{\langle #1|#2\rangle}
\newcommand{\ketbra}[2]{|#1\rangle\!\langle #2|}

\definecolor{dgreen}{rgb}{.1,.5,.1}
\definecolor{grey}{rgb}{.4,.4,.4}

\newcommand{\proj}[1]{\ket{#1}\!\bra{#1}}

\usepackage{textgreek}
\newcommand{\sigp}{\textSigma-protocol\xspace}
\newcommand{\sigps}{\textSigma-protocols\xspace}

\def\id{{\mathds 1}}
\def\PauliX{{\sf X}}

\newcommand{\Lin}{{\cal L}}
\renewcommand{\H}{{\cal H}}

\newcommand{\tr}{\mathrm{tr}}

\newcommand{\C}{\mathbb C}
\newcommand{\Z}{\mathbb Z}

\newcommand{\cnot}{\mathrm{CNOT}}

\newcommand{\noinstance}{\emptyset}
\newcommand{\noinstancesuperscript}{\noinstance}

\def\SE{\mathcal{S}.E}
\def\SRO{\mathcal{S}.RO}

\newcommand{\submission}[2]{#1}

\def\showcomments{1}          
	
\def\serge#1{\ifnum\showcomments=1{\color{red}\sf [SF: #1]}\fi}
\def\jelle#1{\ifnum\showcomments=1{\color{orange}\sf [JD: #1]}\fi}
\def\cs#1{\ifnum\showcomments=1{\color{blue}\sf [CS: #1]}\fi}
\def\cm#1{\ifnum\showcomments=1{\color{dgreen}\sf [CM: #1]}\fi}

\DeclareMathSymbol{\shortminus}{\mathbin}{AMSa}{"39}

\newcommand{\oursubsection}[1]{\subsection{#1}}

\newcommand{\switch}[2]{#1}  

\newcommand{\supmat}{\switch{the appendix}{the suppl. mat.}}

\title{Online-Extractability in the Quantum Random-Oracle Model}

\submission{
\author{Jelle Don\inst{1} \and  Serge Fehr\inst{1,2} \and Christian Majenz\inst{1,4}\and Christian Schaffner \inst{3,4}}
\institute{
	Centrum Wiskunde \& Informatica (CWI), Amsterdam, Netherlands \and 
	Mathematical Institute, Leiden University, Netherlands \and
	Institute for  Logic, Language and Computation, University of Amsterdam, Amsterdam, Netherlands\and
	QuSoft, Amsterdam, Netherlands 
	\\ \email{jelle.don@cwi.nl}, \email{serge.fehr@cwi.nl}, \email{christian.majenz@cwi.nl}, \email{c.schaffner@uva.nl}}
}
{
	\author{\vspace{-1.2cm}}\institute{}
}

\begin{document}
	\maketitle	
	\setcounter{footnote}{0}

	\begin{abstract}
		We show the following generic result. Whenever a quantum query algorithm in the quantum random-oracle model outputs a classical value $t$ that is promised to be in some tight relation with $H(x)$ for some $x$, then $x$ can be efficiently extracted with almost certainty. The extraction is by means of a suitable simulation of the random oracle and works {\em online}, meaning that it is {\em straightline}, i.e., without rewinding, and {\em on-the-fly}, i.e., during the protocol execution and without disturbing it. 
		
		The technical core of our result is a new commutator bound that bounds the operator norm of the commutator of the unitary operator that describes the evolution of the compressed oracle (which is used to simulate the random oracle above) and of the measurement that extracts~$x$. 
		
		We show two applications of our generic online extractability result. We show {\em tight} online extractability of commit-and-open \sigps in the quantum setting, and we offer the first complete post-quantum security proof of the {\em textbook} Fujisaki-Okamoto transformation, i.e, without adjustments to facilitate the proof, including concrete security bounds.  
	\end{abstract}

	\section{Introduction}

	\paragraph{\bf Background. }
	
	\emph{Extractability} plays an important role in cryptography. In an extractable protocol,\submission{ on a high level,}{} an algorithm $\mathcal A$ sends messages that depend on some secret $s$, and while the secret remains private in an honest run of the protocol, an \emph{extractor} can learn $s$ via some form of enhanced access to $\mathcal A$. 
		The probably most prominent example is that of (zero-knowledge) {\em proofs} (or {\em arguments}) {\em of knowledge}, for which, by definition, there must exist an extractor that manages to extract a witness from any successful \submission{yet possibly dishonest prover.}{prover.}
		Another example are {\em extractable commitments}, which have a wide range of applications. 
		Hash-based extractable commitments are extremely simple to construct and prove secure in the random-oracle model (ROM) \cite{Pass03}. Indeed, when the considered hash function $H$ is modelled as a random oracle, the hash input $x$ for the commitment $c =H(x)$, where $x = s\|r$ consists of the actual secret $s$ and randomness $r$,  can be extracted simply by finding a query $x$ to the random oracle that yielded $c$ as an output. 
	
	The general notion of extractability comes in different flavors. The most well-known example is extraction by {\em rewinding}. Here, the extractor is allowed to run $\mathcal A$ several times, on the same private input and using different randomness. This is the notion usually considered in the context of proofs/arguments of knowledge. 
	In some contexts, extraction via rewinding access is not possible. For example, the UC security model prohibits the simulator to rewind the adversary.
	In other occasions, rewinding may be possible but not desirable due to a loss of efficiency, which stems from having to run $\cal A$ multiple times. 
	In comparison, so-called {\em straightline} extraction works with a single ordinary run of $\mathcal A$, without rewinding. Instead, the extractor is then assumed to know some trapdoor information, or it is given enhanced control over some part of the setting. For instance, in the above construction of an extractable commitment, the extractor is given ``read~access'' to $\cal A$'s random-oracle queries. 
		
Another binary criterion is whether the extraction takes place {\em on-the-fly}, i.e., during the run of the protocol, or {\em after-the-fact}, i.e., at the end of the execution. For instance, in the context of proving CCA security for an encryption scheme, to simulate decryption queries without knowing the secret key, it is necessary to extract the plaintext for a queried ciphertext on-the-fly; otherwise, the attacker may abort and not produce the output for which the reduction is waiting. 

The extractability of our running example of an extractable commitment in the ROM is {\em both}, straightline and on-the-fly; we refer to this combination as {\em online} extraction. This is what we are aiming for in this work: online extractability of (general) hash-based commitments, but now with {\em post-quantum security}. 

	For post-quantum security, the ROM needs to be replaced by the {\em quantum} random-oracle model (QROM) \cite{Boneh2011}, to reflect the fact that attackers can implement hash functions on a quantum computer. Here, adversaries have quantum superposition access to the random oracle. Many ROM techniques fail in the QROM due to fundamental features of quantum information, such as the so-called \emph{no-cloning principle}. In particular, it is impossible to maintain a query transcript (a fact sometimes referred to as the \emph{recording barrier}), and so one cannot simply ``search for a query $x$ to the random oracle'', as was exploited for the (classical) RO-security of the extractable-commitment example. 
	
	A promising step in the right direction is the compressed-oracle technique, recently developed by Zhandry~\cite{Zhandry2018}. This technique enables to maintain {\em some sort} of a query transcript, but now in the form of a quantum state. This state can be inspected via quantum measurements, offering the possibility to learn some information about the interaction history of \submission{an algorithm $\cal A$ and }{}the random oracle. However, since quantum measurements disturb the state to which they are applied, and this disturbance is often hard to control, this inspection of the query transcript can {\em per-se}, i.e., without additional argumentation, only be done at the end of the execution (see the {Related Work} paragraph for more on this).

	\paragraph{\bf Our Results.}
	
	Our main contribution is the following generic extractability result in the QROM. We consider an arbitrary quantum query algorithm $\cal A$ in the QROM, which announces during its execution some classical value $t$ that is supposed to be equal to $f(x,H(x))$ for some $x$. Here, $f$ is an arbitrary fixed function, subject to that it must tie $t$ sufficiently to $x$ and $H(x)$, e.g., there must not be too many $y$'s with $f(x,y) = t$; a canonical example is the function $f(x,y) = y$ so that $t$ is supposed to be $t = H(x)$. In general, it is helpful to think of $t = f(x,H(x))$ as a commitment to $x$. We then show that $x$ can be {\em efficiently extracted} with almost certainty. The extraction works {\em online} and is by means of a simulator $\cal S$ that simulates the quantum random oracle, but which additionally offers an {\em extraction interface} that produces a guess $\hat x$ for $x$ when queried with $t$. The simulation is statistically indistiguishable from the real quantum random oracle, and $\hat x$ is such that whenever $\cal A$ outputs $x$ with $f(x,H(x)) = t$ at some later point, $\hat x = x$ except with negligible probability, while $\hat x = \emptyset$ (some special symbol) indicates that $\cal A$ will not be able to output such an~$x$. 
	
	The simulator $\cal S$ simulates the random oracle using Zhandry's compressed-oracle technique, and extraction is done via a suitable measurement of  the compressed oracle's internal register. The technical core of our result is a new 
	bound 
	for
	the operator norm $\|[O,M]\|$ of the commutator of $O$, the unitary operator that describes the evolution of the compressed oracle, and of $M$, the measurement that is used to extract $x$. This commutator bound allows us to show that the extraction measurement disturbs the behavior of the compressed oracle only by a negligible amount, and so can indeed be performed {\em on-the-fly}. 
	At first glance, our technical result has some resemblance with Lemma~39 in~\cite{Zhandry2018}, which also features an almost-commutativity property, and, indeed, with Lemma~\ref{lem:M-commutes-with-local} we use (a reformulated version of) Lemma~39 in~\cite{Zhandry2018}  as a first step in our proof. However, the challenging part of the main proof consists of lifting the almost-commutativity property of the ``local'' projectors $\Pi^x$ from Lemma~\ref{lem:M-commutes-with-local} to the ``global'' measurement $M$  (Lemma~\ref{lem:M-commutes-with-local}). 
	
	We emphasize that even though the existence of the simulator with its extraction interface is proven using the compressed-oracle technique, our presentation is in terms of a black-box simulator $\cal S$ with certain interfaces and with certain promises on its behavior, abstracting away all the (mainly internal) quantum workings. This makes our generic result applicable (e.g.\ for the applications discussed below) without the need to understand the underlying quantum aspects. 
	
	A first concrete application of our generic result is in the context of so-called commit-and-open \sigps. These are (typically honest-verifier zero-knowledge) interactive proofs of a special form, where the prover first announces a list of commitments and is then asked to open a subset of them, 
	chosen at random by the verifier. We show that, when implementing the commitments with a typical hash-based commitment scheme (like committing to $s$ by $H(s\|r)$ with a random $r$), such \sigps allow for {\em online} extraction of a witness in the QROM, with a {\em smaller security loss} than witness extraction via rewinding. 
	
	Equipped with our extractable RO-simulator $\cal S$, the idea for the above online extraction is very simple: we simulate the random oracle using $\cal S$ and use its extraction interface to extract the prover's commitments from the first message of the \sigp. As we work out in detail, this procedure gives rise to an online witness extractor that has a polynomial additive overhead in running time compared to the considered prover, and that outputs a valid witness with a probability that is
	{\em linear} in the difference of the prover's success probability and the trivial cheating probability, up to an additive error
.  Using rewinding techniques, 
on the other hand, incurs a {\em square-root} loss in success probability classically and a {\em cube-root} loss quantumly for special-sound \sigps, and typically an even worse loss in case of weaker soundness guarantees, like a $k$-th-root loss classically and a $(2k+1)$-th-root loss quantumly for $k$-sound protocols. Furthermore, we show that the dominating additive loss of our reduction is necessary in general, due to attacks on the computational binding property of the random-oracle-based commitments. Along the way, we set up a definitional framework for generalized special soundness notions that might be of independent interest. 
	
	A second application of our extractable RO-simulator is a security reduction for the Fujisaki-Okamoto (FO) transformation. We offer the first complete post-quantum security proof of the {\em textbook} FO transformation~\cite{FO99}, with concrete  security bounds. 
		Most of the prior post-quantum security proofs had to adjust the transformation to facilitate the proof (like~\cite{HHK17}); those security proofs either consider a FO variant that employs an {\em implicit-rejection} routine, i.e., where the decapsulation algorithm outputs a pseudo-random key upon an invalid ciphertext rather than a rejection message, or have to resort to an additional ``key confirmation'' hash \cite{TU16} that is appended to the ciphertex, thus increasing the ciphertext size. 
		The {\em unmodified} FO transformation was analyzed in~\cite{Zhandry2018} and~\cite{KKPP20}; 
		however, as we explain in detail in \supmat\ (Sect.~\ref{app:gap}), the given post-quantum security proofs are incomplete, both having the same gap. 

	Beyond its theoretical relevance of showing that no adjustment is necessary\submission{ to admit a post-quantum security proof}{}, the security of the original unmodified FO transformation with explicit rejection in particular ensures that the conservative variant with implicit rejection remains secure even when the decapsulation algorithm is not implemented carefully enough and admits a side-channel attack that reveals information on whether the submitted ciphertext is valid or not. 
	
	The core idea of our proof for the textbook FO transformation is to use the extractability of the RO-simulator to handle the decryption queries. Indeed, letting $f(x,y)$ be the encryption $Enc_{pk}(x;y)$ of the message $x$ under the randomness $y$, a ``commitment'' $t = f(x,H(x))$ is then the encryption of $x$ under the derandomized scheme, and so the extraction interface recovers $x$.

	\paragraph{\bf Related Work.}
		
	The compressed-oracle technique has proven to be a powerful tool for lifting classical ROM proofs to the QROM setting. 
	Examples are \cite{LZ19,CFHL20} for quantum query complexity lower bounds and \cite{HM20} for space-time trade-off bounds, \cite{CMS19} for the security of succinct arguments, \cite{AMRS20} for quantum-access security, and \cite{BHHP19} for a new ``double-sided'' O2H lemma in the context of the FO transformation. 
	In these cases, the argument exploits the possibility to extract information on the interaction history of the algorithm $\cal A$ and the (compressed) oracle {\em after-the-fact}, i.e., at the very end of the run. 
	

In addition, some tools have been developed that allow measuring (the internal state of) the compressed oracle {\em on-the-fly}, which then causes the state, and thus the behavior of the oracle, to change. 
In some cases, the disturbance is significant yet asymptotically good enough for the considered application, causing ``only'' a polynomial blow-up of a negligible error term, as, e.g., in \cite{LZ19a} 
	for proving the security of the Fiat-Shamir transformation. In other cases \cite{Zhandry2018,CMSZ19}, it is shown for some limited settings that certain measurements do not render the simulation of the random oracle distinguishable (except for negligible advantage). The indifferentiability result in \cite{CMSZ19}, for example, only uses measurements that have an almost certain outcome.  

	In particular, \cite{Zhandry2018} contains a security reduction for the Fujisaki-Okamoto (FO) transformation that implicitly uses a measurement similar to the one we analyze in Section~\ref{sec:ComBound}, but without analyzing the disturbance it causes. We discuss this in more detail in \supmat\ (Sect.~\ref{app:gap}). The same gap exists in recent follow-up work by Katsumata, Kwiatkowski, Pintore and Prest~\cite{KKPP20}, who follow the FO proof outline from~\cite{Zhandry2018}.


	\section{Preliminaries}
	
	For Sect.~\ref{sec:ComBound} and \ref{sec:generic} (only), we assume some familiarity with the mathematics of quantum information as well as with the compressed-oracle technique of~\cite{Zhandry2018}. Below, we summarize the concepts that will be of particular importance. 
	For a 
	function or algorithm $f$, we 
	write $\mathrm{Time}[f]$ to denote the time complexity of  (an algorithm computing) $f$. 

	\oursubsection{Mathematical Preliminaries}
        
        Let $\H$ be a finite-dimensional complex Hilbert space. We use the standard bra-ket notation for the vectors in $\H$ and its dual space. We write $\|\ket{\varphi}\|$ for the (Euclidean) norm $\|\ket{\varphi}\| = \sqrt{\braket{\varphi}{\varphi}}$ of $\ket{\varphi} \in \H$. Furthermore, for an operator $A \in \Lin(\H)$, we denote by $\|A\|$ its {\em operator norm}, i.e., $\|A\| = \max_{\ket{\psi}} \| A \ket{\psi}\|$, where the max is over all $\ket{\psi} \in \H$ with norm $1$. We assume the reader to be familiar with basic properties of these norms, like triangle inequality, $\|\ketbra{\varphi}{\psi}\| = \|\ket{\varphi}\| \|\ket{\psi}\|$, $\|A\ket{\varphi}\| \leq \|A\| \|\ket{\varphi}\|$, $\|AB\| \leq \|A\| \|B\|$, etc. 
        Less well known may be the inequality%
        \footnote{It is immediate for normalized $\ket\phi$ and $\ket\psi$  when expanding both vectors in an orthonormal basis containing $\ket\varphi$ and $\frac{\ket \psi-\braket{\varphi}{\psi}\ket\varphi}{\sqrt{1-|\braket{\varphi}{\psi}|^2}}$, and the general case then follows by homogeneity of the norms. }
       \begin{equation}\label{eq:norminequality}
       \|\ketbra{\varphi}{\psi}-\ketbra{\psi}{\varphi}\| \le \|\ket{\varphi}\| \|\ket{\psi}\| \, .
       \end{equation} 
Another basic yet important property that we will exploit is the following. 
        
	\begin{lemma}\label{lem:blockdiag-opnorm}
		Let $A$ and $B$ be operators in $\Lin(\H)$ with $A^\dagger B = 0$ \switch{(i.e., they have orthogonal images)}{}and $A B^\dagger = 0$\switch{ (i.e., they have orthogonal supports)}{}. Then, $\| A + B \| \leq \max \{ \| A \| , \| B \|\}$.  
	\end{lemma}
%
	Exploiting that $\|A \otimes B\| = \|A\| \|B\|$, the following is a direct consequence\switch{ of Lemma~\ref{lem:blockdiag-opnorm}}{}. 
	
	\begin{corollary}\label{cor:control-opnorm}
		If $A = \sum_x \proj{x} \otimes A^x$\switch{, i.e., $A$ is a controlled operator,} then $\|A\| \leq \max_x \|A^x\|$.
	\end{corollary}

	\begin{definition}
	For \switch{operators }{}$A,B \in \Lin(\H)$, the {\em commutator} is \switch{defined as }{}$[A,B] := AB - BA$. 
	\end{definition}
	Some obvious properties of the commutator are: 
	\switch{
	\begin{align}\label{eq:ComBasics}
	[B,A] = -[A,B] = [A,\id-B]  && \quad\text{and}&& [A \otimes \id, B \otimes C] = [A,B] \otimes C \, ,
	\end{align}
	as well as
	\begin{align}\label{eq:CommutatorOfProduct}
	 [AB,C] = A[B,C]+[A,C]B
	\end{align}}{
\begin{align}\label{eq:ComBasics}
	[B,A] = -[A,B] &= [A,\id-B] \; ,  \quad
	 [A \otimes \id, B \otimes C] = [A,B] \otimes C \\[1ex]
	&\text{and }\qquad[AB,C] = A[B,C]+[A,C]B\, .\label{eq:CommutatorOfProduct}
\end{align}
}
	Combining the right equality in (\ref{eq:ComBasics}) with basic properties of the operator norm, if $\|C\| \leq 1$, e.g., if $C$ is a unitary of a projection, we have
	\begin{equation}\label{eq:ComOfTensorProduct}
	\| [A \otimes \id, B \otimes C]  \| =\|  [A,B] \|  \| C \|  \leq \|  [A,B] \| \, .
	\end{equation}

	It is common in quantum information science to write $A_X$ to emphasize that the operator $A$ acts on {\em register} $X$, i.e., on a Hilbert space $\H_X$ that is labeled by the \switch{letter/symbol}{} $X$. It is then understood that when applied to registers $X$ and $Y$, say, $A_X$ acts as $A$ on register $X$ and as identity $\id$ on register $Y$, i.e., $A_X$ is identified with $A_X \otimes \id_Y$. Property (\ref{eq:ComOfTensorProduct}) would then e.g.\ be written as $\| [A_X, B_X \otimes C_Y]\|  \leq \|  [A_X,B_X] \|$. 
	In this work, we will write or not write these subscripts emphasizing the register(s) at our convenience; typically we write them when the argument crucially depends on the registers, and we may omit them otherwise.

	Another important matrix norm is the {\em Schatten-1} or {\em trace norm}, $\|A\|_1=\tr\bigl[\sqrt{A^\dagger A}\bigr]$. For density matrices $\rho$ and $\sigma$, the {\em trace distance} is then defined as $\delta(\rho,\sigma)=\frac 1 2\|\rho-\sigma\|_1$. By equation (9.110) in \cite{Nielsen:2011:QCQ:1972505} and a short calculation, any norm-$1$ vectors $\ket{\varphi}$ and $\ket\psi$ satsify 
\begin{equation}\label{eq:norminequality2}
\delta(\proj\varphi, \proj\psi)\le \|\ket\varphi -\ket{\psi}\| \, .
\end{equation}
	For probability distributions $p$ and $q$, we write $\delta(p,q)$ for the {\em total variational distance}; this is justified as $\|\rho_0-\rho_1\|_1=\delta(p_0,q_1)$ for $\rho_i = \sum_x p_i(x) \proj{x}$, $i=0,1$. In case of a hybrid classical-quantum state, consisting of  a randomized classical value $x$ that follows a distribution $p$ and of a quantum register $W$ with a state $\rho_W^x$ that depends on $x$, we write $[x,W]= \sum_x p(x) \proj{x} \otimes \rho_W^x$
	.\footnote{In this equality and at other occasions, we use the same letter, here $x$, for the considered {\em random variable} as well as for a {\em particular value}. }
	When the distribution $p$ and the density operators $\rho_W^x$ are implicitly given by a game (or experiment) $\cal G$ then we may write $[x,W]_{\cal G}$, in particular when considering and comparing different such games. For instance, we write $\delta\bigl([x,W]_{\cal G},[x,W]_{{\cal G}'}\bigr)$ for the trace distance of  the respective density matrices in game $\cal G$ and in game ${\cal G}'$.

	\oursubsection{The (Compressed) Random Oracle}\label{subsec:oracle}
	
	\paragraph{\bf The (quantum) random-oracle model.}
	
In the {\em random-oracle model}, a cryptographic hash function $H: \mathcal X \to\mathcal Y$ is treated as an external oracle $RO$ that the adversary needs to query on $x \in \cal X$ in order to learn $H(x)$. The random oracle answers these queries by means of a uniformly random function $H: \mathcal X \to\mathcal Y$. 	
For concreteness, we restrict here to ${\cal Y} = \{0,1\}^n$; on the other hand, we do not further specify the domain $\mathcal{X}$ except that we assume it to have an efficiently computable order, so one may well think of $\cal X$ as ${\cal X} = \{1,\ldots,M\}$ for some positive $M \in \Z$ or as bit strings of bounded size. 
We then often write $RO(x)$ instead of $H(x)$ in order to emphasize that $H(x)$ is obtained by querying the random oracle and/or to emphasize the randomized nature of $H$. 
\submission{

In the {\em quantum} random oracle model (QROM), a quantum algorithm $\cal A$ may make {\em superposition queries} to $RO$, meaning that the oracle acts as unitary $\ket{x}\ket{y} \mapsto \ket{x}\ket{y \oplus H(x)}$.  
}{In the {\em quantum} random oracle model (QROM), a quantum algorithm $\cal A$ may make {\em superposition queries} to $RO$, meaning that the oracle acts as unitary $\ket{x}\ket{y} \mapsto \ket{x}\ket{y \oplus H(x)}$.  }	 
The QROM still admits {\em classical} queries, which are queries with the query register set to $\ket{x}\ket{0}$ for some $x$, and the second register is subsequently measured to obtain the classical output $y$.

	\paragraph{\bf The compressed oracle.}


We recall here (some version of) the {\em compressed} oracle, as introduced in \cite{Zhandry2018}, which offers a powerful tool for QROM proofs. 
For this purpose, we consider the multi-register $D = (D_x)_{x \in \cal X}$, where the state space of $D_x$ is given by $\H_{D_x} = \C[\{0,1\}^n \cup \{\bot\}]$, meaning that it is spanned by an orthonormal set of vectors $\ket{y}$ labelled by $y \in \{0,1\}^n \cup \{\bot\}$. 
	The initial state is set to be $\ket{\boldsymbol \bot}_D := \bigotimes_{x}\ket{\bot}_{D_x}$. Consider the unitary $F$ defined by 
	\begin{align*}
	F\ket\bot=\ket{\phi_0} \; ,\quad F\ket{\phi_0}=\ket \bot \quad\text{and}\quad
	F\ket{\phi_y}=\ket{\phi_y} \;\, \forall \, y \in \{0,1\}^n \setminus \{0^n\} \, ,  
	\end{align*}
	where $\ket{\phi_y} := H \ket y$ with  $H$ the \submission{Walsh-}{}Hadamard transform on $\C[\{0,1\}^n] = (\C^2)^{\otimes n}$. Exploiting the relation $\ket y=2^{-n/2} \sum_{\eta} (-1)^{\eta \cdot y} \ket{\phi_\eta}$, we see that
		\begin{equation}\label{eq:F}
		F\ket y=\ket y+2^{-n/2}\left(\ket\bot\!-\!\ket{\phi_0}\right) \, . 
		\end{equation}
	When the oracle is queried, a unitary $O_{XYD}$, acting on the query registers $X$ and $Y$ and the oracle register $D$, is applied, given by
	\submission{\begin{equation*}
	O_{XYD}=\sum_x\proj{x}_X\otimes O^x_{YD_x},
	\end{equation*}
	with 
	\begin{equation}
	O^x_{YD_x} = F _{D_x}\cnot_{Y D_x}F_{D_x} \label{eq:defO}
      \end{equation}}{
  \begin{align}
  	O_{XYD}=\sum_x\proj{x}_X\otimes O^x_{YD_x} &&\text{ with } &&	O^x_{YD_x}= F _{D_x}\cnot_{Y D_x}F_{D_x} \, , \label{eq:defO}
  \end{align}
}
      where $\cnot\submission{_{YD_x}}{} \ket{y} \ket{y_x}\submission{}{\!} =\submission{}{\!}\ket{y \oplus y_x} \ket{y_x}$ for $y, y_x \in \{0,1\}^n$  and acts as identity on $\ket{y}\ket{\bot}$
      

As long as no other operations are applied to the state of $D$, this compressed oracle is perfectly indistinguishable from the quantum random oracle. Also, the support of the state of $D_x$ then remains orthogonal to $\ket{\phi_0}$ for any $x$. However, these properties may change when, e.g., measurements are performed on~$D$. The oracle may then behave differently than the quantum random oracle, and the state of $D$ may then have a non-trivial overlap with $\ket{\phi_0}$. We note that, by the convention on $\cnot$ to act trivially when the control register is in state $\ket{\bot}$, it holds that $O^x_{YD_x} \ket{y} \ket{\phi_0} =\ket{y} \ket{\phi_0}$.

When considering a {\em classical} query, which is a query with the $XY$-register in state $\ket{x}\ket{0}$ for some $x$, it is understood that the $Y$-register is then measured after the application of $O_{XYD}$. If $D_x$ is in state $\rho$ then a classical query on $x$ will give response $h$ with probability $\tr(\proj{h} F\rho F)$\,---\,unless $\rho$ has nontrivial overlap with $\ket{\phi_0}$ and $h=0$, in which a classical query on $x$ will give response $0$ with probability $\tr(\proj{0}F\rho F)+\tr(\proj{\bot}F\rho F)$. We note that, for any $h \in \cal Y$ and $\rho = \proj{h}$, 
	\begin{align}
	\tr(\proj{h} F\rho& F) = |\bra{h}F\ket{h}|^2 = \Big|\bra{h} \Bigl(\ket{h} + \textstyle 2^{-n/2}(\ket{\bot} - \ket{\phi_0})\Big)\Big|^2 \nonumber\\
	&= \Big| 1 - 2^{-n/2}\braket{h}{\phi_0} \Big|^2 = \Big| 1 - 2^{-n} \Big|^2 \geq 1 - 2\cdot 2^{-n} \, .\label{Eq:MeasureClassicalQuery}
	\end{align}
	Vice-versa, after a classical query on $x$ with response $h$, the state of $D_x$ is $F\ket{h}$\,---\,unless, the state of $D_x$ prior to the query had a nontrivial overlap with $\ket{\phi_0}$ and $h=0$, in this case, the state after the query is supported by $F\ket{0}$ and $F\ket{\bot} = \ket{\phi_0}$.

	\paragraph{\bf Efficient representation of the compressed oracle. }\label{subsec:compressed}

\def\Enc{\mathsf{SparseEnc}}

Following \cite{Zhandry2018}, one can make the (above variant of the) compressed oracle efficient. Indeed, by applying the standard classical sparse encoding to quantum states with the right choice of basis, one can {\em efficiently} maintain the state $D$, compute the unitary $O_{XYD}$, and extract information from $D$. 
More details are given in \supmat\ (Sect.~\ref{subsec:compressed}). 
\submission{For simplicity, we will express things in the remainder of the paper in terms of the inefficient variant of the compressed oracle, but we stress that by the said means all relevant unitaries and measurements can be efficiently computed. }{For simplicity, we will use the inefficient variant of the compressed oracle in this paper, and analyze the efficient variant only when needed.}

	\section{Main Technical Result: A Commutator Bound}\label{sec:ComBound}
	
	\submission{
	Our main technical result is a bound on the operator norm of the commutator $[O_{XYD},M_{DP}]$ of the unitary $O_{XYD}$, which describes the evolution of the compressed oracle, and the (purified) measurement $M_{DP}$. Informally, this measurement checks if there is a pair $(x,y)$ in the database satisfying a given relation. If yes, it outputs (the smallest such) $x$, otherwise it outputs $\noinstance$. 
	A small bound on this commutator means that performing this measurement during the runtime of an oracle algorithm $\cal A$ interacting with a (compressed
	) random oracle, has little effect. 
}{}

	\oursubsection{Setup and the Technical Statement}
	
	Throughout this section, we consider an arbitrary but fixed relation  $R\subset\mathcal X\times \{0,1\}^n$. A crucial parameter of the relation $R$ is the number of $y$'s that fulfill the relation together with $x$, maximized over all possible $x \in \mathcal{X}$:  
	\begin{equation}\label{eq:Gamma_R}
	\Gamma_R := \max_{x \in \cal X}\left|\left\{y\in\{0,1\}^n\big|(x,y)\in R\right\}\right| \, .
	\end{equation}
	Given the relation $R$, we consider the following projectors: 
	\begin{equation}\label{eq:measurement}
	\Pi_{D_x}^{x} := \!\sum_{y \text{ s.t.} \atop (x,y)\in R}\! \proj{y}_{D_x}
	\quad\text{and}\quad
	\Pi_D^{\noinstancesuperscript} := \id_D - \sum_{x \in \cal X} \Pi_{D_x}^{x}  = \bigotimes_{x \in \cal X} \bar \Pi^x_{D_x}  
	\end{equation}
	with $\bar \Pi^x_{D_x} := \id_{D_x} - \Pi^x_{D_x}$. 
Informally, $\Pi^x_{D_x}$ checks whether register $D_x$ contains a value $y \neq \bot$ such that $(x,y)\in R$.  We then define the measurement $\mathcal M = \mathcal M^R$ to be given by the  projectors  
	\begin{equation}\label{eq:basicmeas}
	\Sigma^x := \bigotimes_{x' < x} {\bar \Pi^{x'}}_{D_{x'}} \otimes \Pi^x_{D_x} \
	\quad\text{and}\quad
	\Sigma^\noinstancesuperscript := \id - \sum_{x'} \Sigma^{x'} =  \bigotimes_{x'} {\bar \Pi^{x'}}_{D_{x'}} = \Pi^{\noinstancesuperscript}  
	\end{equation}
where $x$ ranges over all $x \in \cal X$. 
	Informally, a measurement outcome $x$ means that register $D_x$ is the first that contains a value $y$ such that $(x,y)\in R$
	; outcome $\noinstance$ means that no register contains such a value. 
	For technical reasons, we consider the {\em purified} measurement $M_{DP} = M_{DP}^R \in \Lin(\H_D \otimes \H_{R})$ 
	given by the unitary\footnote{Both in $\PauliX^x$ and in $w + x$ we understand $x \in {\cal X} \cup \{\noinstance\}$ to be encoded as an element in $\Z/(|{\cal X}|\!+\!1)\Z$, $\dim(\H_P) = d:= |{\cal X}|+1$, and $\PauliX \in \Lin(\H_P)$ is the generalized Pauli of order $d$ that maps $\ket{w}$ to $\ket{w + 1}$.}
	\begin{equation}\label{eq:purified-measurement}
	M_{DP} := \!\sum_{x \in {\cal X} \cup \{\noinstance\}}\!\! \Sigma^x \otimes \PauliX^x: \ket{\varphi}_D\ket{w}_P \mapsto \sum_{x \in {\cal X} \cup \{\noinstance\}} \Sigma^x \ket{\varphi}_D \ket{w + x}_P \, .
	\end{equation}
	The following main technical result is a bound on the norm of \switch{the commutator }{}$[O_{XYD},M_{DP}]$.

	\begin{theorem}\label{thm:commutator}
		For any relation $R\subset \mathcal X\times \{0,1\}^n$ and $\Gamma_R$ as defined in Eq.~\eqref{eq:Gamma_R}, the purified measurement $M_{DP}$ defined in Eq.~\eqref{eq:purified-measurement} almost commutes with the oracle unitary $O_{XYD}$: 
		$$
		\bigl\| \, [O_{XYD},M_{DP}] \, \bigr\| \leq  8 \cdot 2^{-n/2} \sqrt{2\Gamma_R} \, .
		$$
	\end{theorem}
We note that Lemma 8 in \cite{CMS19} (with the subsequent discussion there) also provides a bound on \submission{the norm of }{}a commutator involving $O_{XYD}$; however, there are various differences that make the two bounds incomparable. E.g., we consider a specific {\em measurement} whereas Lemma~8 in \cite{CMS19} is for a rather general {\em projector}. See further down for a comparison with Lemma~39 in~\cite{Zhandry2018}.


	\begin{corollary}\label{cor:commutator}
	For any state vector $\ket{\psi} \in \H_{WXYDP}$, with $W$ an arbitrary additional register, \switch{the state vectors }{}$\ket{\psi'} := O_{XYD} M_{DP}\ket{\psi}$ and $\ket{\psi''} := M_{DP} O_{XYD} \ket{\psi}$ satisfy
	$$
	\delta\bigl(\proj{\psi'}, \proj{\psi''}\bigr) \leq 8 \cdot 2^{-n/2} \sqrt{2\Gamma_R} \, . 
	$$
	The same holds for mixed states $\rho' := O_{XYD} M_{DP} \rho M_{DP}^\dagger O_{XYD}^\dagger$ and $\rho'' := M_{DP} O_{XYD} \rho O_{XYD} ^\dagger M_{DP}^\dagger$.  
	\end{corollary}
	
	\begin{proof}
	By elementary properties and applying Theorem~\ref{thm:commutator}, we have that
	$$
	\big\|\ket{\psi'} - \ket{\psi''} \big\| \switch{= \big\|  (O_{XYD} M_{DP} - M_{DP} O_{XYD})\ket{\psi} \big\|}{} \leq \big\| [O_{XYD}, M_{DP}] \big\| \leq 8 \cdot 2^{-n/2} \sqrt{2\Gamma_R} \, ,
	$$
	and the claim on the trace distance then follows from (\ref{eq:norminequality2}). The claim for mixed states follows from purification. 
	\qed
	\end{proof}

	\oursubsection{The Proof}

	We prove the Theorem~\ref{thm:commutator} by means of the following two lemmas. 
	
	\begin{lemma}\label{lem:simple}
	Let $F$ and $O^x_{YD_x}$ be the unitaries introduced in Sect.~\ref{subsec:oracle}, and let $\Pi_{D_x}^x$ and $\Pi^{\noinstancesuperscript}_{D}$ be as in (\ref{eq:measurement}). Set $\Gamma_{x} := \left|\left\{y\in\{0,1\}^n\big|(x,y)\in R\right\}\right|$. Then
	\begin{align*}
	\left\|\left[F_{D_x}, \Pi_{D_x}^{x}\right]\right\| \le 2^{-n/2}\sqrt{2\Gamma_x} \, , \qquad \text{as well as} \qquad\qquad\qquad \\[1ex]
	\left\|\left[O^x_{YD_x},\Pi^{x}_{D_x} \right]\right\| \le 2\cdot 2^{-n/2}\sqrt{2\Gamma_{x}} 
	\quad\text{and}\quad
	\big\|\big[O^x_{YD_x},\Pi^{\noinstancesuperscript}_{D} \big]\big\| \le 2\cdot 2^{-n/2}\sqrt{2\Gamma_{x}}  \, . 
	\end{align*}
	\end{lemma}
%
The bound on $\|[F, \Pi^{x}]\|$ can be considered a compact reformulation of (a variant of) Lemma~39 in~\cite{Zhandry2018}. We state it here in this form, and (re-)prove it in \supmat~(Sect. \ref{sec:SuppProofs}), for convenience and  completeness. The conceptually new and technically challenging ingredient to the proof of Theorem~\ref{thm:commutator} is Lemma~\ref{lem:M-commutes-with-local} below.%
\footnote{The challenging aspect of Lemma~\ref{lem:M-commutes-with-local} is that $M_{DP}$ is made up of an exponential number of projectors $\Pi^x$, and thus the obvious approach of using triangle inequality leads to an exponential blow-up of the error term. Naively, one might hope to avoid the exponential blow-up (at the cost of introducing a blow-up linear in the number of prior queries) by using the efficient representation of the compressed oracle  (as discussed in Sect.~\ref{subsec:compressed} in \supmat.); however, the two representations are isometrically equivalent, and so switching the representation has no effect in that respect. }


	\begin{lemma}\label{lem:M-commutes-with-local}
		The purified measurement $M_{DP}$ defined in Equation \eqref{eq:purified-measurement} satisfies
		\begin{align*}
		\big\| [F_{D_x},M_{DP}] \big\| &\leq   3 \big\| [F_{D_x},\Pi_D^x] \big\| + \big\| [F_{D_x},\Pi_D^{\noinstancesuperscript}] \big\| && \text{and} \\[1ex]
		\big\| [O^x_{Y D_x},M_{DP}] \big\| &\leq   3 \big\| [O^x_{Y D_x},\Pi_D^x] \big\| + \big\| [O^x_{Y D_x},\Pi_D^{\noinstancesuperscript}] \big\| \, . 
		\end{align*}
	\end{lemma}
	
	\begin{proof}
	We do the proof for the second claim. The first is proven exactly the same way: the sole property we exploit from $O^x_{Y D_x}$ is that it acts only on the $D_x$ register within $D$, which holds for $F_{D_x}$ as well.  
          Let $$\bar\Delta^\xi :=  \bigotimes_{\xi' < \xi} {\bar \Pi^{\xi'}}_{D_{\xi'}}$$ be the projection that accepts if no register $D_{\xi'}$ with $\xi' < \xi$ contains a value $y'$ with
          $(\xi',y') \in R$, and let $\Delta^\xi$ be the complement. 
		We then have, using that  $\Pi^\xi$ and $\bar\Delta^\xi$ act on disjoint registers, 
		\begin{equation}\label{eq:Sigma2Pi}
		\Sigma^\xi = \bar\Delta^\xi \otimes \Pi^\xi = \Pi^\xi  \bar\Delta^\xi  =  \bar\Delta^\xi \Pi^\xi  \, .
		\end{equation}
		We also observe that, with respect to the Loewner order, $\bar\Delta^{\xi'} \geq \bar\Delta^\xi$ for $\xi' < \xi$. 
		Taking it as understood that $O^x_{Y D_x}$ acts on registers $Y$ and $D_x$, we can write
		\begin{equation}\label{eq:defofMpluggedin}
		[O^x,M_{DP}]  = \sum_\xi  [O^x,\Sigma^\xi] \otimes  \PauliX^\xi +  [O^x,\Sigma^\noinstancesuperscript] \otimes  \PauliX^\noinstancesuperscript \, . 
		\end{equation}
		Exploiting basic properties of the operator norm and recalling that $\Sigma^\noinstancesuperscript = \Pi_D^{\noinstancesuperscript}$, we see that the norm of the last term is bounded by $\| [O^x,\Sigma^{\noinstancesuperscript}] \| = \| [O^x,\Pi^{\noinstancesuperscript}] \|$. 

To deal with the sum in~\eqref{eq:defofMpluggedin}, we use $\id = \Delta^\xi + \bar\Delta^\xi$ to further decompose
		\begin{equation}\label{eq:idsandwich}
		[O^x,\Sigma^\xi]  = \bar\Delta^\xi [O^x,\Sigma^\xi] \bar\Delta^\xi  + \bar\Delta^\xi O^x,\Sigma^\xi] \Delta^\xi  +  \Delta^\xi [O^x,\Sigma^\xi] \bar\Delta^\xi  +  \Delta^\xi [O^x,\Sigma^\xi] \Delta^\xi \, . 
		\end{equation}
		We now analyze the four different terms. For the first one, using (\ref{eq:Sigma2Pi}) we see that
		\begin{align*}
		\submission{\bar\Delta^\xi [&O^x,\Sigma^\xi] \bar\Delta^\xi =  \bar\Delta^\xi \big(O^x\Sigma^\xi - \Sigma^\xi O^x \big) \bar\Delta^\xi 
		= \bar\Delta^\xi O^x\Pi^\xi  \bar\Delta^\xi -  \bar\Delta^\xi \Pi^\xi O^x \bar\Delta^\xi =  \bar\Delta^\xi [O^x,\Pi^\xi] \bar\Delta^\xi \, ,}
	{\!\bar\Delta^\xi [&O^x\!\!,\!\Sigma^\xi] \bar\Delta^\xi\!=\!  \bar\Delta^\xi \big(O^x\Sigma^\xi\!-\! \Sigma^\xi O^x \big) \bar\Delta^\xi 
		\!=\! \bar\Delta^\xi O^x\Pi^\xi \! \bar\Delta^\xi \!-\!  \bar\Delta^\xi \Pi^\xi O^x \!\bar\Delta^\xi\! =\!  \bar\Delta^\xi [O^x\!\!,\!\Pi^\xi] \bar\Delta^\xi \, ,}
		\end{align*}
		which vanishes for $\xi \neq x$, since then $O^x$ and $\Pi^\xi$ act on different registers and thus commute. For $\xi=x$, its norm is upper bounded by $\|[O^x,\Pi^x]\|$. 
		
		We now consider the second term; the third one can be treated the same way by symmetry, and the fourth one vanishes, as will become clear immediately from below. Using (\ref{eq:Sigma2Pi}) and $\bar\Delta^\xi\Delta^\xi = 0$, so that $\bar\Delta^\xi\Sigma^\xi = 0$, we have
		\begin{equation}\label{eq:N}
		\bar\Delta^\xi [O^x,\Sigma^\xi] \Delta^\xi =  \bar\Delta^\xi \big(O^x\Sigma^\xi - \Sigma^\xi O^x \big) \Delta^\xi =  \Sigma^\xi O^x \Delta^\xi =: N_\xi \, .
		\end{equation}
    		Looking at (\ref{eq:defofMpluggedin}), we want to control the norm of the sum $N := \sum_\xi N_\xi \otimes X^\xi$. 		
To this end, we show that $N_{\xi}$ and \smash{$N_{\xi'}$} have orthogonal images and orthogonal support, i.e., \smash{$N_{\xi'}^\dagger N_{\xi} = 0 = N_{\xi'} N_{\xi}^\dagger$}, for all $\xi \neq \xi'$. We first observe that if $x \geq \xi$ then $O^x$ commutes with $\Delta^\xi$, since they act on different registers then, and thus 
$$
N_{\xi} =  \Sigma^\xi O^x \Delta^\xi =  \Sigma^\xi \Delta^\xi O^x  =  \Pi^\xi \bar\Delta^\xi \Delta^\xi O^x = 0 \, ,
$$
exploiting once more that $\bar\Delta^\xi \Delta^\xi = 0$. 
Therefore, we only need to consider $N_\xi,N_{\xi'}$ for $\xi,\xi' > x$ (see Fig.~\ref{fig:N} top left), where we may assume $\xi > \xi'$. For the orthogonality of the images, we observe that 
\begin{equation}\label{eq:core}
\Pi^{\xi'} \bar\Delta^{\xi}  = 0
\end{equation} 
by definition of \smash{$\bar\Delta^{\xi}$} as a tensor product with $\bar \Pi^{\xi'}$ being one of the components. Therefore,
$$
(\Sigma^{\xi'})^\dagger \Sigma^{\xi} = \Sigma^{\xi'} \Sigma^{\xi}  = \bar\Delta^{\xi'} \Pi^{\xi'} \bar\Delta^{\xi} \Pi^{\xi} = 0 \, ,
$$
and $N_{\xi'}^\dagger N_{\xi} = 0$ follows directly (see also Fig.~\ref{fig:N} top right). For the orthogonality of the supports, we recall that \smash{$\bar\Delta^{\xi'} \geq \bar\Delta^\xi$}, and thus \smash{$\Delta^{\xi'} \leq \Delta^\xi$}, from which it follows that \smash{$\Delta^{\xi} \Delta^{\xi'} = \Delta^{\xi'}$}. \smash{$N_{\xi'} N_{\xi}^\dagger = 0$} then follows by exploiting (\ref{eq:core}) again (see Fig.~\ref{fig:N} bottom). 
			
		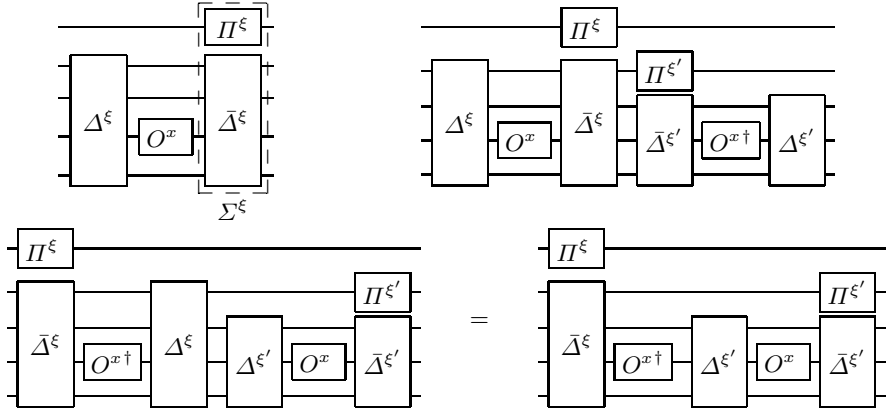
\begin{figure}[h]
		$$ 
		\Qcircuit @C=0.5em @R=.4em {                      
			&  \qw                                 & \qw                      & \gate{\,\Pi^\xi\,}                           & \qw  \\
			&  \multigate{3}{\Delta^\xi}  & \qw                      &  \multigate{3}{\bar\Delta^\xi}       & \qw  \\
			& \ghost{\Delta^\xi}             & \qw                     &\ghost{\bar \Delta^\xi}               & \qw         \\
			& \ghost{\Delta^\xi}             &  \gate{{O^x}\phantom{^\dagger\!}}             & \ghost{\bar \Delta^\xi}                 & \qw  \\
			& \ghost{\Delta^\xi}             &  \qw                    & \ghost{\bar \Delta^\xi}              & \qw  \gategroup{1}{4}{5}{4}{.5em}{--} \\
			&                                         &                            & \raisebox{-4ex}{$\Sigma^\xi$}     
		}
		\qquad\qquad\qquad
		\Qcircuit @C=0.4em @R=.2em {
			&  \qw                                 & \qw                                           & \gate{\Pi^{\xi\phantom{'}}}           & \qw   & \qw                                    & \qw            & \qw                                      & \qw  \\
			&  \multigate{3}{\Delta^{\xi}}  & \qw                                         &  \multigate{3}{\bar\Delta^{\xi}}   & \qw    & \gate{\Pi^{\xi'}}            & \qw           &  \qw                                      & \qw  \\
			& \ghost{\Delta^{\xi}}             & \qw                                        &\ghost{\bar \Delta^{\xi}}               & \qw      & \multigate{2}{\!\bar\Delta^{\xi'}\!}    & \qw           &   \multigate{2}{\!\Delta^{\xi'}\!}    & \qw         \\
			& \ghost{\Delta^{\xi}}             &  \gate{{O^x}\phantom{^\dagger\!}}       & \ghost{\bar \Delta^{\xi}}        & \qw    &  \ghost{\!\bar \Delta^{\xi'}\!}          &  \gate{{O^x}^\dagger}     &  \ghost{\!\Delta^{\xi'}\!}             & \qw  \\
			& \ghost{\Delta^{\xi}}             &  \qw                                         & \ghost{\bar \Delta^{\xi}}              & \qw   &  \ghost{\!\bar \Delta^{\xi'}\!}        & \qw             & \ghost{\!\Delta^{\xi'}\!}              & \qw   
		}
		$$
		
		\medskip
		
		$$ 
	\Qcircuit @C=0.4em @R=.2em {
			&  \gate{\Pi^{\xi\phantom{'}}}         & \qw                      & \qw                                      & \qw  & \qw                                        & \qw                                       & \qw                                      & \qw  \\
			&  \multigate{3}{\bar\Delta^{\xi}}  & \qw                        &  \multigate{3}{\Delta^{\xi}}   & \qw  & \qw                                        & \qw                                        & \gate{\Pi^{\xi'}}                       & \qw  \\
			& \ghost{\bar\Delta^{\xi}}             & \qw                        &\ghost{\Delta^{\xi}}               & \qw  & \multigate{2}{\!\Delta^{\xi'}\!}    & \qw                                      &   \multigate{2}{\!\bar\Delta^{\xi'}\!}    & \qw         \\
			& \ghost{\bar\Delta^{\xi}}             &  \gate{{O^x}^\dagger}    & \ghost{\Delta^{\xi}}             & \qw  &  \ghost{\!\Delta^{\xi'}\!}          &  \gate{{O^x}\phantom{^\dagger\!}}     &  \ghost{\!\bar\Delta^{\xi'}\!}         & \qw  \\
			& \ghost{\bar\Delta^{\xi}}             &  \qw                       & \ghost{\Delta^{\xi}}             & \qw  &  \ghost{\!\Delta^{\xi'}\!}          & \qw                                           &\ghost{\!\bar\Delta^{\xi'}\!}     & \qw  
		}
		\qquad \raisebox{-7ex}{$=$} \qquad	
	\Qcircuit @C=0.4em @R=.2em {
			&  \gate{\Pi^{\xi\phantom{'}}}         & \qw                        & \qw  & \qw                                        & \qw                                       & \qw                                      & \qw  \\
			&  \multigate{3}{\bar\Delta^{\xi}}  & \qw                         & \qw  & \qw                                        & \qw                                        & \gate{\Pi^{\xi'}}                       & \qw  \\
			& \ghost{\bar\Delta^{\xi}}             & \qw                         & \qw  & \multigate{2}{\!\Delta^{\xi'}\!}    & \qw                                      &   \multigate{2}{\!\bar\Delta^{\xi'}\!}    & \qw         \\
			& \ghost{\bar\Delta^{\xi}}             &  \gate{{O^x}^\dagger}     & \qw  &  \ghost{\!\Delta^{\xi'}\!}          &  \gate{{O^x}\phantom{^\dagger\!}}     &  \ghost{\!\bar\Delta^{\xi'}\!}         & \qw  \\
			& \ghost{\bar\Delta^{\xi}}             &  \qw                        & \qw  &  \ghost{\!\Delta^{\xi'}\!}          & \qw                                           &\ghost{\!\bar\Delta^{\xi'}\!}     & \qw  
		}
		$$
		\caption{\switch{The o}{O}perators $N_\xi$ (top left), $N_{\xi'}^\dagger N_{\xi}$ (top right), and $N_{\xi'} N_{\xi}^\dagger$ (bottom), for $x < \xi' < \xi$. }\label{fig:N}
	\end{figure}

		These orthogonality properties for the images and supports of the $N_\xi$ immediately extend to $N_\xi \otimes X^\xi$, so we have
		$$
		\|N\| \leq \max_{\xi > x} \| N_\xi \otimes \PauliX^\xi \|  \leq \max_{\xi > x} \| N_\xi \| 
		$$
by Lemma \ref{lem:blockdiag-opnorm}. 
		Recall from (\ref{eq:N}) that $N_\xi = \bar\Delta^\xi [\Sigma^\xi, O^x] \Delta^\xi$.
		Furthermore, we exploit that, by definition, $\Sigma^\xi$ is in tensor-product form and $O^x$ acts trivially on all components in this tensor product except for the component $\bar\Pi^x$, so that $[\Sigma^\xi ,O^x ] = [\bar\Pi^x ,O^x ]$ by property (\ref{eq:ComOfTensorProduct}). Thus, 
		\begin{align*}
			\| N_\xi \| \leq \|[\Sigma^\xi , O^x ]\| = \|[\bar\Pi^x, O^x ]\| \nonumber 
			= \|[\Pi^x, O^x ]\| \, . 
		\end{align*}
		Using the triangle inequality with respect to the sum versus the last term in \eqref{eq:defofMpluggedin}, and another triangle inequality with respect to the decomposition \eqref{eq:idsandwich}, we obtain 
		the claimed inequality. 
		\qed
	\end{proof}
	The proof of Theorem \ref{thm:commutator} is now an easy consequence.
	\begin{proof}[of Theorem \ref{thm:commutator}]
		Since $O_{XYD}$ is a control unitary $O_{XYD} = \sum_{x} \proj{x} \otimes O^{x}_{Y D_x}$, controlled by $\ket{x}$, while $M_{DP}$ does not act on register $X$, it follows that 
		$$
		\big\| [O_{XYD},M_{DP}] \big\| \leq \max_{x} \big\| [O^{x}_{Y D_x},M_{DP}] \big\| \, .
		$$
		The claim \switch{of the theorem }{}now follows by combining Lemma \ref{lem:M-commutes-with-local}  with Lemma~\ref{lem:simple}. 
		\qed
	\end{proof}
	
\oursubsection{A First Immediate Application}

As an immediate application of the commutator bound of Theorem \ref{thm:commutator}, we can easily derive the following generic query-complexity bound for finding $x$ with $(x,H(x)) \in R$ and $\Gamma_R$ as defined in Eq.~\eqref{eq:Gamma_R}. \submission{}{Applied to $R = {\cal X} \times \{0^n\}$, where $\Gamma_R = 1$, we recover the  famous lower bound for search in a random function. }

\begin{proposition}\label{prop:Grover}
	For any algorithm $\mathcal A$ that makes $q$ queries to the random oracle $RO$, 
	\begin{equation}
			\Pr_{x\leftarrow\mathcal A^{RO}}\left[ \big(x,RO(x)\big) \in R \right]\le 152 (q+1)^2 \Gamma_R/2^{n} \, .
	\end{equation}
\end{proposition}

\begin{proof}
	Consider the modified algorithm $\mathcal A'$ that runs $\mathcal A$ to obtain output $x$, makes a query to obtain $RO(x)$ and outputs $(x,RO(x))$.  By Lemma 5 in \cite{Zhandry2018}, we have that%
	\footnote{Lemma 5 in \cite{Zhandry2018} applies to an algorithm $\cal A$ that outputs both $x$ and what is supposed to be its hash value; this is why we need to do this additional query. }
	\begin{equation}
		\sqrt{\Pr_{x\leftarrow\mathcal {A'}^H}\left[ (x,RO(x)) \in R \right]}\le \sqrt{\Pr_{x' \leftarrow G^R}\left[x'\neq\emptyset \, \right]}+2^{-n/2},
	\end{equation}
where $G^R$ is the following procedure/game:
	(1) run $\mathcal A'$ using the compressed oracle, and
	(2) apply the measurement ${\cal M}^R$ to obtain $x' \in {\cal X} \cup \{\emptyset\}$, which is the same as preparing a register $P$, applying $M_{DP}=M^R_{DP}$, and measuring $P$.

In other words, writing $\ket{\psi}_{WXY}$ for the initial state of $\mathcal A'$ and $V_{WXY}$ for the unitary applied between any two queries of $\mathcal A'$(which we may assume to be fixed without loss of generality), and setting $U_{WXYD} := V_{WXY}O_{XYD}$, $\Pi_P := \mathds 1_P-\proj\emptyset_P$ and \smash{$\ket{\Psi} := \ket\psi_{WXY} \otimes \ket{\bot}^{\otimes|\mathcal X|}_D \otimes \ket{0}_P$}, we have, omitting register subscripts, 
\begin{align*}
	&\sqrt{\Pr\left[x'\neq\emptyset \, \right]} = \big\|\Pi M U^{q+1}\ket\Psi \big\|
	\le \sum_{i=1}^{q+1} \big\|\Pi U^{i-1} [M,U] U^{q+1-i}  \ket\Psi \big\|  + \big\|\Pi U^{q+1}M\ket\Psi \big\| \\
	&\;\leq (q+1) \big\|[M_{DP},O_{XYD}] \big\| + \big\|  \Pi_P M_{DP}\ket\Psi \big\| 
	= (q+1)\big\|[M_{DP},O_{XYD}] \big\| 
	\leq 8 \cdot 2^{-n/2} (q+1) \sqrt{2\Gamma_R} \, ,
\end{align*}
where the last equation exploits that $\Pi_P M_{DP}$ applied to $\ket{\bot}^{\otimes|\mathcal X|}_D \otimes \ket{0}_P$ vanishes, and the final inequality is by Theorem~\ref{thm:commutator}. 
Observing $(8\sqrt{2} + 1)^2 = 129+16\sqrt{2} \approx 151.6$ finishes the proof. \qed
\end{proof}
\submission{Applied to $R = {\cal X} \times \{0^n\}$, where $\Gamma_R = 1$, we recover the  famous lower bound for search in a random function. In essence, our commutator bound replaces the ``progress-measure'' argument in the search-lower-bound proof from \cite{Zhandry2018}.
\begin{corollary}
	For any algorithm $\mathcal A$ that makes $q$ queries to the random oracle\switch{ $RO$}{}, 
	\begin{equation}
			\Pr_{x\leftarrow\mathcal A^{RO}}\left[RO(x)=0^n\right]\le 152 (q+1)^2/2^{n}.
	\end{equation}
\end{corollary}
}{}

	\section{Extraction of Random-Oracle Based Commitments}\label{sec:generic}

Throughout this Sect.~\ref{sec:generic}, let $f: \mathcal X \times \mathcal Y \to \mathcal T$ be an arbitrary fixed function with ${\cal Y} = \{0,1\}^n$. For a hash function $H : \mathcal X \to \mathcal Y$, which will then be modelled as a random oracle $RO$, we will think and sometimes speak of $f(x,H(x))$ as a {\em commitment} of $x$ (though we do not require it to be a commitment scheme in the strict sense). Typical examples are $f(x,y) = y$ and $f(x,y) = {\sf Enc}_{pk}(x;y)$, where the latter is the encryption of $x$ under public key $pk$ with randomness~$y$.

\oursubsection{Informal Problem Description}

Consider a query algorithm ${\cal A}^{RO}$ in the random oracle model, which, during the course of its run, announces some $t \in \cal T$. This $t$ is supposed to be $t = f(x,RO(x))$ for some $x$, and, indeed, ${\cal A}^{RO}$ may possibly reveal $x$ later on\submission{, i.e., open the commitment}. 
Intuitively, in order for the required relation between $x$ and $t$ to hold, we expect that ${\cal A}^{RO}$ {\em first} has to query $RO$ on $x$ and only {\em then} can output $t$; thus, one may hope to be able to extract $x$ from $RO$ {\em early on}, i.e., at the time ${\cal A}^{RO}$ announces $t$. 

This is clearly true when $\cal A$ is restricted to classical queries, simply by checking all the queries made so far. This observation was first made and utilized by Pass~\cite{Pass03} and only requires looking at the query transcript (it can be done in the \emph{non-programmable} ROM). As the extractor does not change the course of the experiment, \submission{ it is in particular also suitable in situations where it is necessary to extract an opening on the fly, i.e., while guaranteeing that $\cal A$ still proceeds to produce its output (e.g. for multiple-committer parallel extraction \cite{ABGKM20}).}{it works on-the-fly.}

In the setting considered here, ${\cal A}^{RO}$ may query the random oracle in {\em superposition} over various choices of~$x$, making it impossible to maintain a classical query transcript. On the positive side, since the output $t$ is required to be classical, ${\cal A}^{RO}$ has to perform a measurement before announcing~$t$, enforcing such a superposition to collapse.%
\footnote{We can also think of this measurement being done by 
the interface that receives~$t$. 
}
We show here that early extraction of $x$ is indeed possible in this quantum setting as well.

Note that if the goal is to extract {\em the same} $x$ as $\cal A$ will (potentially) output, which is what we aim for, then we must naturally assume that it is hard for $\cal A$ to find $x \neq x'$ that are both consistent with the same $t$, i.e., we must assume the commitment to be binding. 
Formally, \submission{for the upcoming discussion in this section to be meaningful,} we will think of $\Gamma(f)$ and $\Gamma'(f)$, defined as follows, to be small compared to \submission{$|{\cal Y}| = 2^n$}{$2^n$}. When $f$ is fixed, we simply write $\Gamma$ and~$\Gamma'$. 

\begin{definition}\label{def:GammaAndGPrime} For $f: \mathcal X \times \{0,1\}^n \to \mathcal T$, we define
$$
\Gamma(f) := \max_{x,t} | \{y \mid f(x,y)= t \} | \ \ \text{and}\ \   \Gamma'(f) := \max_{x \neq x' , y'} | \{y \mid f(x,y)= f(x',y') \} | \, .
$$
\end{definition}
For the example $f(x,y) = y$, we have $\Gamma(f) = 1 =  \Gamma'(f)$. For the example $f(x,y) = {\sf Enc}_{pk}(x;y)$, they both depend on the choice of the encryption scheme but typically are small, e.g. $\Gamma(f) = 1$ if $\sf Enc$ is injective as a function of the randomness $y$ and $\Gamma'(f) = 0$ if there are no decryption errors.  

\submission{
\begin{remark}
We note that the ratio $\Gamma(f)/2^n$ remains unaffected when $n$ is increased, i.e., if $\tilde n \geq n$ and $\tilde f: \mathcal X \times \{0,1\}^{\tilde n} \to \mathcal T$ is given by $\tilde f(x,y\|y') := f(x,y)$ for all $x \in \cal X$, $y \in \{0,1\}^n$ and $y' \in \{0,1\}^{\tilde n-n}$, then $\Gamma(\tilde f)/2^{\tilde n} = \Gamma(f)/2^n$, because the additional $\tilde{n} -n$ bits of $y'$ do not affect the conditions on $\tilde f$ in Definition~\ref{def:GammaAndGPrime}, so both numerator and denominator of the fraction get multiplied by $2^{\tilde n -n}$. The same holds for $\Gamma'(f)/2^n$. 
\end{remark}
}{}

\oursubsection{The Extractable RO-Simulator $\cal S$}

Towards formalizing the above goal, we introduce a simulator $\cal S$ that replaces $RO$ and tries to extract $x$ early on, right after ${\cal A}$ announces $t$. In more detail, $\cal S$ acts as a black-box oracle with two interfaces, the {\em RO-interface} $\SRO$ providing access to the simulated random oracle, and the {\em extraction interface} $\SE$ providing the functionality to extract $x$ early on (see Fig.~\ref{fig:S}, left). In principle, both interfaces can be accessed quantumly, i.e., in superposition over different classical inputs, but in our applications we only use classical access to $\SE$. 
We stress that $\cal S$ is per-se {\em stateful} and thus may change its behavior from query to query. 

Formally, the considered simulator ${\cal S}$ is defined to work as follows. It simulates the random oracle and answers queries to $\SRO$ by means of the compressed oracle. For the $\SE$ interface, upon a classical input $t \in \cal T$, ${\cal S}$ applies the measurement ${\cal M}^t := {\cal M}^{R_t}$ from \eqref{eq:basicmeas} for the relation $R_t := \{(x,y)\,|\,f(x,y) = t\}$ to obtain $\hat x \in {\cal X} \cup \{\emptyset\}$, which it then outputs (see Fig.~\ref{fig:Sdef}). In case of a {\em quantum} query to $\SE$, the above is performed coherently: given the query registers $TP$, the unitary $\sum_t \proj{t}_T \otimes M^{R_t}_{DP}$ is applied to $TPD$, and \submission{registers $TP$ are}{$TP$ is} then returned. 

\begin{figure}[h]

\begin{center} \makebox[\textwidth][c]{ \fbox{ 
\begin{minipage}[t]{\linewidth}
The extractable RO-oracle $\cal S$:
\begin{description}\vspace{-1ex}\setlength{\parskip}{0.5ex}
\item[\it Initialization:] $\cal S$ prepares its internal register $D$ to be in state $\ket{\boldsymbol \bot}_D := \bigotimes_{x}\ket{\bot}_{D_x}$. 
\item[$\SRO$\it -query:] Upon a (quantum) RO-query, with query registers $XY$, $\cal S$ applies $O_{XYD}$ to registers $XYD$. 
\item[$\SE$\it -query:] Upon a classical extraction-query with input $t$, $\cal S$ applies ${\cal M}^t$ to $D$ and returns the outcome~$\hat x$. 
\end{description}
\end{minipage}
}}
\end{center}
\vspace{-2ex}
\caption{The (inefficient version of the) simulator $\cal S$, restricted to classical extraction queries.}\label{fig:Sdef}
\end{figure}

We note that, as described here, the simulator $\cal S$ is inefficient, having to maintain an exponential number of qubits; however, using the sparse representation of the internal state $D$, as discussed in \supmat, Sect.~\ref{subsec:compressed}, $\cal S$ can well be made efficient without affecting its query-behavior (see Theorem \ref{thm:MainFeatures} for details). 

The following statement captures the core properties of $\cal S$. 
We refer to two subsequent queries as being {\em independent} if they can in principle be performed in either order, i.e., if the input to one query does not depend on the output of the other. More formally, e.g., two $\SRO$ queries are independent if they can be captured by first preparing the two in-/output registers $XY$ and $X'Y'$, and then doing the two respective queries with $XY$ and $X'Y'$. The commutativity claim then means that the order does not matter. 
Furthermore, whenever we speak of a {\em classical} query (to $\SRO$ or to $\SE$), we consider the obvious classical variant of the considered query, with a classical input and a classical response. 
Finally, 
the almost commutativity claims are in terms of the trace distance of the (possibly quantum) output of any algorithm interacting with $\cal S$ arbitrarily and doing the two  considered independent queries in one or the other order. 

\begin{theorem}\label{thm:MainFeatures}
The extractable RO-simulator $\cal S$ constructed above, with interfaces $\SRO$ and $\SE$, satisfies the following properties. 
\begin{itemize}\vspace{-1ex}\setlength{\parskip}{0.5ex}
\item[1.\!] If $\SE$ is unused, $\mathcal S$ is perfectly indistinguishable from the random oracle $RO$.  \\[-2ex]
\item[2.a] Any two subsequent independent queries to $\SRO$ 
commute. \submission{In particular}{Thus},  two subsequent {\em classical} $\SRO$-queries with the same input $x$ give identical responses. 
\item[2.b] Any two subsequent independent queries to  $\SE$ 
commute. \submission{In particular}{Thus},  two subsequent {\em classical} $\SE$-queries with the same input $t$ give identical responses. 
\item[2.c] Any two subsequent independent queries to $\SE$ and $\SRO$ 
$8\sqrt{2\Gamma(f)/2^n}$-almost-commute. \\[-2ex]
\item[3.a] Any classical query $\SRO(x)$ is idempotent.\footnote{I.e., applying it twice \submission{in a row }{}has the same effect on the state of $\cal S$ as applying it once.} 
\item[3.b] Any classical query $\SE(t)$ is idempotent.   \\[-2ex] 
\item[4.a] If $\hat x = \SE(t)$ and $\hat h = \SRO(\hat x)$ are two subsequent classical queries then 
\begin{equation}\label{eq:extract-then-RO} 
	\Pr[f(\hat x, \hat h) \neq t\wedge\hat x \neq\emptyset] \le \Pr[f(\hat x, \hat h) \neq t\mid\hat x \neq\emptyset] \le 2\cdot 2^{-n}\Gamma(f)
\end{equation} 
\item[4.b] If $h = \SRO(x)$ and $\hat x = \SE(f(x,h))$ are two subsequent classical queries such that no prior query to $\SE$ has been made, then 
\begin{equation}
	\Pr[\hat x =\emptyset]\le 2 \cdot 2^{-n}.
\end{equation}

\end{itemize}
Furthermore, the total runtime of $\mathcal S$, when implemented using the sparse representation of the compressed oracle described in Sect. \ref{subsec:compressed}, is bounded as
$$
T_{\mathcal S}= O\bigl(q_{RO} \cdot q_E\cdot \mathrm{Time}[f] + q_{RO}^2\bigr) \, ,
$$
where $q_E$ and $q_{RO}$ are the number of queries to $\SE$ and $\SRO$, respectively
.
\end{theorem}


\begin{proof}
All the properties follow rather directly by construction of $\cal S$. 
Indeed, without $\SE$-queries, $\cal S$ is simply the compressed oracle, known to be perfectly indistinguishable from the random oracle, confirming 1. Property 2.a follows from the fact that the unitaries $O_{XYD}$ and $O_{X'Y'D}$, acting on the same register $D$ but on distinct query registers, are both controlled unitaries with control register $D$, conjugated by a fixed unitary ($F^{\otimes |\mathcal X|}$). They thus commute. For 2.b, the claim follows from the fact that the unitaries $M^t_{DP}$ and $M^{t'}_{DP'}$ commute, as they are both controlled unitaries with control register $D$. 2.c is a direct consequence of our main technical result Theorem~\ref{thm:commutator} (in the form of Cor.~\ref{cor:commutator}). 3.a follows from the fact that a classical $\SRO$ query with input $x$ acts as a projective measurement on register $D_x$, which is, as any projective measurement, idempotent. Thus, so is the measurement ${\cal M}^t$, confirming 3.b.

To prove 4.a, consider the state $\rho_{D_{\hat x}}$ of register $D_{\hat x}$ after the measurement ${\cal M}^{t}$ that is performed by the extraction query $\hat x = \SE(t)$, assuming $\hat x \neq \emptyset$.  Let $\ket{\psi}$ be a purification of $\rho_{D_{\hat x}}$. 
By definition of ${\cal M}^{t}$, it holds that $\Pi^{\hat x}_{D_{\hat x}}\ket{\psi} = \ket{\psi}$. Then, understanding that all operators act on register $D_{\hat x}$, by definition of $\bar \Pi^{\hat x}$ the probability of interest is bounded as%
\footnote{The first inequality is an artefact of the $\proj{\bot}$-term in $\bar \Pi^{\hat x}$ contributing to the probability of $\hat h = 0$, as discussed in Sect.~\ref{subsec:oracle}. }
\begin{align*}
	\Pr[f(\hat x, \hat h) \neq t \,|\, \hat x \neq\emptyset] &\leq \big\| \bar \Pi^{\hat x} F \ket{\psi} \big\|^2 = \big\| \bar \Pi^{\hat x} F \Pi^{\hat x}\ket{\psi} \big\|^2 \leq  \big\| \bar \Pi^{\hat x} F \Pi^{\hat x} \big\|^2 \leq  \big\|[F,\Pi^{\hat x}]\big\|^2 \, ,
\end{align*}
where the last inequality exploits that $\bar \Pi^{\hat x} \Pi^{\hat x}  = 0$. The claim now follows from Lemma \ref{lem:simple}.

For 4.b, we first observe that, given that there were no prior extraction queries, the state of $D_x$ before the $h = \SRO(x)$ query has no overlap with $\ket{\phi_0}$, and thus the state after the query is $F\ket{h}$ (see the discussion above Equation \eqref{Eq:MeasureClassicalQuery}). For the purpose of the argument, instead of applying the measurement ${\cal M}^{f(x,h)}$ to answer the $\SE(f(x,h))$ query, we may equivalently consider a measurement in the basis $\{\ket{\bf y}\}$,  and then set $\hat x$ to be the smallest element $\cal X$ so that $f(\hat x,y_{\hat x}) = t := f(x,h)$, with $\hat x = \emptyset$ if no such element exists. Then, 
$$
\Pr[\hat x \neq \emptyset] = \Pr[\exists \, \xi: f(\xi,y_\xi) = t ] \geq  \Pr[f(x,y_x) = t]  \geq  \Pr[y_x = h] = |\bra{h}F\ket{h}|^2 \geq 1 - 2\cdot 2^{-n}
$$
where the last two (in)equalities are by Equation (\ref{Eq:MeasureClassicalQuery}). 

%
\qed
\end{proof}

\oursubsection{Two More Properties of $\cal S$}

On top of the above basic features of our extractable RO-simulator $\cal S$, we show the following two additional, more technical, properties, which in essence capture that the extraction interface cannot be used to bypass query hardness results.

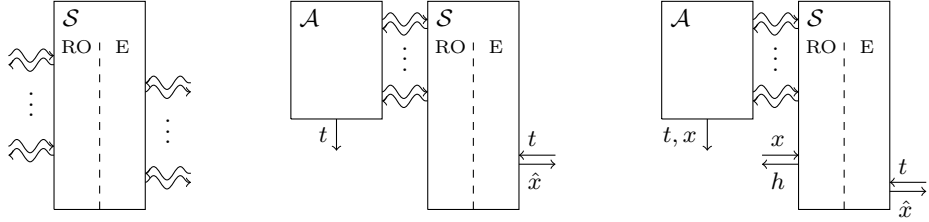
\begin{figure}[h]
$$
\begin{tikzpicture}[scale=0.12, baseline=0]
\draw (15,-8) -- (15,15) node[anchor=north west] {$\cal S$} -- (25,15)  -- (25,-8) -- (15,-8) ;
\draw [dashed] (20,-8) -- (20,11);
\node at (17.5,10)  {\scriptsize RO};
\node at (22.5,10)  {\scriptsize E};
\draw [->,snake it] (10,9) -- (15,9);
\draw [<-,snake it] (10,8) -- (15,8);
\node at (12.5,4.5)  {$\vdots$};
\draw [->,snake it] (10,-1) -- (15,-1);
\draw [<-,snake it] (10,-2) -- (15,-2);
\draw [<-,snake it] (25,6) -- (30,6);
\draw [->,snake it] (25,5) -- (30,5);
\node at (27.5,1.5)  {$\vdots$};
\draw [<-,snake it] (25,-4) -- (30,-4);
\draw [->,snake it] (25,-5) -- (30,-5);
\end{tikzpicture}
\qquad\qquad
\begin{tikzpicture}[scale=0.12, baseline=0]
\draw (0,2) -- (0,15) node[anchor=north west] {$\cal A$} -- (10,15)  -- (10,2) -- (0,2) ;
\draw [->,snake it] (10,13) -- (15,13);
\draw [<-,snake it] (10,12) -- (15,12);
\node at (12.5,9.5)  {$\vdots$};
\draw [->,snake it] (10,5) -- (15,5);
\draw [<-,snake it] (10,4) -- (15,4);
\draw [->] (5,2)  -- (5,-1.5) node[anchor=south east] {$t$};
\draw (15,-8) -- (15,15) node[anchor=north west] {$\cal S$} -- (25,15)  -- (25,-8) -- (15,-8) ;
\draw [dashed] (20,-8) -- (20,11);
\node at (17.5,10)  {\scriptsize RO};
\node at (22.5,10)  {\scriptsize E};
\draw [<-] (25,-2) node[anchor=south west] {$t$} -- (29,-2);
\draw [->] (25,-3) node[anchor=north west] {$\hat x$} -- (29,-3);
\end{tikzpicture}
\qquad\qquad
\begin{tikzpicture}[scale=0.12, baseline=0]
\draw (0,2) -- (0,15) node[anchor=north west] {$\cal A$} -- (10,15)  -- (10,2) -- (0,2) ;
\draw [->,snake it] (10,13) -- (15,13);
\draw [<-,snake it] (10,12) -- (15,12);
\node at (12.5,9.5)  {$\vdots$};
\draw [->,snake it] (10,5) -- (15,5);
\draw [<-,snake it] (10,4) -- (15,4);
\draw [->] (5,2)  -- (5,-1.5) ;
\node at (2,0) {$t,x$};
\draw (15,-8) -- (15,15) node[anchor=north west] {$\cal S$} -- (25,15)  -- (25,-8) -- (15,-8) ;
\draw [dashed] (20,-8) -- (20,11);
\node at (17.5,10)  {\scriptsize RO};
\node at (22.5,10)  {\scriptsize E};
\draw [->] (11,-2) node[anchor=south west] {$x$} -- (15,-2);
\draw [<-] (11,-3) node[anchor=north west] {$h$} -- (15,-3);
\draw [<-] (25,-5) node[anchor=south west] {$t$} -- (29,-5);
\draw [->] (25,-6) node[anchor=north west] {$\hat x$} -- (29,-6);
\end{tikzpicture}
\vspace{-2ex}
$$
\caption{The extractable RO-simulator $\cal S$, with its $\SRO$ and $\SE$ interfaces, distinguished here by queries from the left and right (left), and the games considered in Prop.~\ref{prop:multiHardProperty} (middle) and \ref{prop:HardCollision} (right) for  $\ell = 1$.  Waved arrows denote quantum queries, straight arrows denote classical queries.}\label{fig:S}
\end{figure}

The first property is easiest to understand in the context of the example $f(x,y) = y$, where $\SE(t)$ tries to extract a hash-preimage of $t$, and where the relations $R$ and $R'$ in Prop.~\ref{prop:multiHardProperty} below then coincide. 
In this case, recall from Prop.~\ref{prop:Grover} that, informally, if $\Gamma_R$ is small then it is hard to find $x \in \cal X$ so that $t:= RO(x)$ satisfies $(x,t) \in R$. The statement below ensures that this hardness cannot be bypassed by first selecting a ``good'' hash value $t$ 
and then trying to extract a preimage by means of $\SE$ (Fig.~\ref{fig:S}, middle).
For instance, setting $t:= t_\circ$ for a given target $t_\circ$ and extracting $\hat x :=  \SE(t)$, we cannot hope for $\hat x$ to satisfy $\SRO(\hat x) = t$; unless there was a prior query to $\SRO$ with response $t_\circ$, the extraction will provide $\hat x = \emptyset$ most likely. 



	\begin{proposition}\label{prop:multiHardProperty}
		Let $R' \subseteq {\cal X} \times {\cal T}$ be a relation. Consider a query algorithm $\cal A$ that makes $q$ queries to the $\SRO$ interface of $\cal S$ but no query to $\SE$, outputting some ${\bf t} \in \cal T^\ell$. For each $i$, let $\hat x_i$ then be obtained by making an additional query to $\SE$ on input $t_i$ (see Fig.~\ref{fig:S}, middle). Then 
		$$
		\Pr_{{\bf t} \leftarrow {\cal A}^{\SRO} \atop {\hat x}_i \leftarrow \SE(t_i)}[\,\exists \, i: ({\hat x}_i, t_i) \in R' ] \leq 128 \cdot q^2 \Gamma_R / 2^n \, ,
		$$
		where $R \subseteq {\cal X} \times {\cal Y}$ is the relation $(x,y) \in R \, \Leftrightarrow (x,f(x,y)) \in R'$ and $\Gamma_R$ as in~(\ref{eq:Gamma_R}).
	\end{proposition}

\begin{proof}
	The considered experiment is like the experiment $G^R$ in the proof of Prop.~\ref{prop:Grover}, the only difference being that in $G^R$ the measurement ${\cal M}^R$ is applied to register $D$ to obtain $x'$ (see Fig.~\ref{fig:CircuitsHardProperty}, middle), while here we have $\ell$ measurements ${\cal M}^{t_{i}}$ that are applied to obtain ${\hat x}_{i}$ (see Fig.~\ref{fig:CircuitsHardProperty}, left). Since all measurements are defined by means of projections that are diagonal in the same basis $\{\ket{\bf y}\}$ with  $\ket{\bf y}$ ranging over ${\bf y} \in ({\cal Y} \cup \{\bot\})^{\cal X}$, we may equivalently measure $D$ in that basis to obtain ${\bf y}$  (see Fig.~\ref{fig:CircuitsHardProperty}, right), and let $\hat x_{i}$ be minimal so that $f({\hat x}_{i},y_{\hat x_{i}}) = t_{i}$ (and $\hat x_{i} = \emptyset$ if no such value exists), and let $x'$ be minimal so that $(x',y_{x'}) \in R$ (and $x' = \emptyset$ if no such value exists). 
	By the respective definitions of ${\cal M}^t_{i}$ and ${\cal M}^R$, both pairs of random variables $({\bf \hat x}, {\bf t})$ and $(x', {\bf t})$ then have the same distributions as in the respective original two games. But now, we can consider their joint distribution and argue that 
	\begin{align*}
	\Pr[\,\exists \, i &: (\hat x_{i}, t_{i}) \in R'  ] = \Pr[\,\exists \, i: (\hat x_{i}, f(\hat x_{i},y_{{\hat x}_{i}})) \in R' ] \\ &= \Pr[\, \exists\, i: (\hat x_{i}, y_{{\hat x}_{i}}) \in R ] 
	\leq  \Pr[\,\exists \, x : (x, y_x) \in R ]  = \Pr[x'\!\neq\emptyset] \, .
	\end{align*}
	The bound on $\Pr[x'\!\neq\emptyset]$ from the proof of Prop.~\ref{prop:Grover} concludes the proof. \qed
\end{proof}

	\begin{figure}[h]		
		$$
		\Qcircuit @C=0.4em @R=.25em @!R {
			&\lstick{D}  & \qw                            &  \multigate{2}{O}     &  \qw & \push{\!...\!\!}  &  &    \qw                                 & \qw                &    \measuretab{\!{\cal M}^t\!}  &  \rstick{\hat x}\cw  \\
			&\lstick{X}  &  \multigate{2}{\!A_0\!}    &  \ghost{O}               & \qw & \push{\!...\!\!}   & & \multigate{2}{\!A_q\!}   \\
			&\lstick{Y}  &  \ghost{\!A_0\!}               &  \ghost{O}               & \qw  & \push{\!...\!\!}  & &  \ghost{\!A_q\!}         \\
			&                &   \ghost{\!A_0\!}             & \qw                        & \qw  & \push{\!...\!\!}  &  &   \ghost{\!A_q\!}                &   \push{\,\;t\;\,}\cw  & \cctrl{-3}      &  \rstick{t}\cw     \\
		}
		\qquad \qquad
		\Qcircuit @C=0.4em @R=.25em @!R {
			& \push{\!...\!\!}  &  &    \qw                                    &    \measuretab{\!{\cal M}^R\!}  &  \rstick{x'}\cw  \\
			& \push{\!...\!\!}   & & \multigate{2}{\!A_q\!}   \\
			& \push{\!...\!\!}  & &  \ghost{\!A_q\!}         \\
			& \push{\!...\!\!}  &  &   \ghost{\!A_q\!}                &   \push{\,\;t\;\,}\cw    \\
		}
		\qquad \qquad
		\Qcircuit @C=0.4em @R=.25em @!R {
			& \push{\!...\!\!}  &  &    \qw                                 &   \meter &  \rstick{{\bf y} \,\leadsto\, \hat x, x'}\cw  \\
			& \push{\!...\!\!}   & & \multigate{2}{\!A_q\!}   \\
			& \push{\!...\!\!}  & &  \ghost{\!A_q\!}         \\
			& \push{\!...\!\!}  &  &   \ghost{\!A_q\!}                &   \push{\,\;t\;\,}\cw    \\
		}\qquad\qquad
		\vspace{-2ex}
		$$
		\caption{Quantum circuit diagrams for the experiments in the proof of Prop.~\ref{prop:multiHardProperty} for the case $\ell=1$.}\label{fig:CircuitsHardProperty}
	\end{figure}
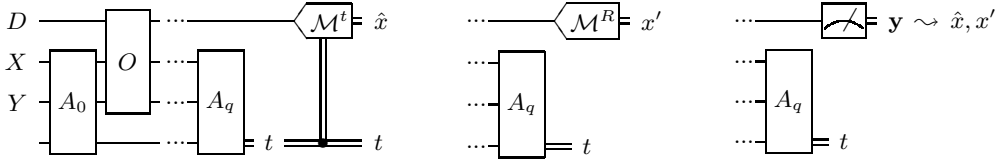

In a somewhat similar spirit, the following ensures that if it is hard in the QROM to find $x$ and $x'$ with $f(x,RO(x)) = f(x',RO(x'))$ then this hardness cannot be bypassed by, say, first choosing $x$, querying $h = \SRO(x)$, computing $t:= f(x,h)$, and then extracting $\hat x := \SE(t)$. The latter will most likely give $\hat x = x$, except, intuitively, if $\SRO$ has additionally been queried on a colliding~$x'$. 

\begin{proposition}\label{prop:HardCollision}
Consider a query algorithm $\cal A$ that makes $q$ queries to $\SRO$ but no query to $\SE$, outputting some $t \in \cal T$ and $x \in \cal X$. Let $h$ then be obtained by making an additional query to $\SRO$ on input $x$, and $\hat x$ by making an additional query to $\SE$ on input $t$ (see Fig.~\ref{fig:S}, right). Then
$$
\Pr_{\tiny\begin{array}{c}t,x \leftarrow {\cal A}^{\SRO} \\ h \leftarrow \SRO(x) \\ \hat x \leftarrow \SE(t)\end{array}}[ \hat x \neq x  \wedge f(x,h) = t ] \leq 
\frac{40 e^2 (q+2)^3 \Gamma'(f)+2}{2^n} \, .
$$
More generally, if $\cal A$ outputs $\ell$-tuples  ${\bf t} \in {\cal T}^\ell$ and ${\bf x} \in {\cal X}^\ell$, and ${\bf h} \in {\cal Y}^\ell$ is obtained by querying $\SRO$ component-wise on $\bf x$, and ${\bf\hat x} \in ({\cal X} \cup \{\emptyset\})^\ell$ by querying $\SE$ component-wise on $\bf t$, then
$$
\Pr_{\tiny\begin{array}{c}{\bf t},{\bf x} \leftarrow {\cal A}^{\SRO} \\ {\bf h} \leftarrow \SRO({\bf x}) \\ {\bf\hat x} \leftarrow \SE({\bf t})\end{array}}[ \exists \, i: \hat x_i \neq x_i  \wedge f(x_i,h_i) = t ]\leq \frac{40 e^2 (q+\ell+1)^3 \Gamma'(f)+2}{2^n} \, .
$$
\end{proposition}
The proof is similar in spirit to the proof of Prop.~\ref{prop:multiHardProperty}, but relying on the hardness of collision finding (Lemma~\ref{lem:ColBound}) rather than on (the proof of) Prop.~\ref{prop:Grover}, and so is moved to \supmat\ (Sect.~\ref{sec:SuppProofs}).

\begin{remark}\label{rem:HardCollision}\sloppy
The claim of Prop.~\ref{prop:HardCollision} stays true when the queries $\SRO(x_i)$ are not performed as {\em additional} queries {\em after} the run of $\cal A$ but are explicitly {\em among} the $q$ queries that are performed by $\cal A$ {\em during} its run. \submission{One way to see this is to use 2.a and 3.a of Theorem~\ref{thm:MainFeatures} to re-do these queries once more after the run of $\cal A$, which does not affect the subsequent $\SE$-queries. Alternatively,}{Indeed} we observe that the proof does not exploit that these queries are performed at the end, which additionally shows that in this case the $\ell$-term on the right hand side of the bound vanishes, i.e., scales as $(q+1)^3$ rather than as $(q+\ell+1)^3$ .  
\end{remark}

\oursubsection{Early Extraction}

We consider here the following concrete setting. Let $\cal A$ be a two-round query algorithm, interacting with the random oracle $RO$ and behaving as follows. At the end of the first round, ${\cal A}^{RO}$ outputs some $t \in \cal T$, and at the end of the second round, it outputs some $x \in \cal X$ that is supposed to satisfy $f(x,RO(x)) = t$; on top, ${\cal A}^{RO}$ may have some additional (possibly quantum) output $W$ (see Fig.~\ref{fig:AandS},~left). 

We now show how the extractable RO-simulator $\cal S$ provides the means to extract $x$ early on, i.e., right after $\cal A$ has announced $t$. To formalize this claim, we consider the following experiment, which we denote by $G_{\mathcal S}^{\mathcal A}$. The RO-interface $\SRO$ of $\mathcal S$ is used to answer all the oracle queries made by ${\cal A}$. In addition, as soon as $\cal A$ outputs $t$, the interface $\SE$ is queried on $t$ to obtain $\hat x \in \cal X \cup \{\emptyset\}$, and after $\mathcal A$ has finished, $\SRO$ is queried on $\mathcal A$'s final output $x$ to generate $h$; see Fig.~\ref{fig:AandS} (right).

	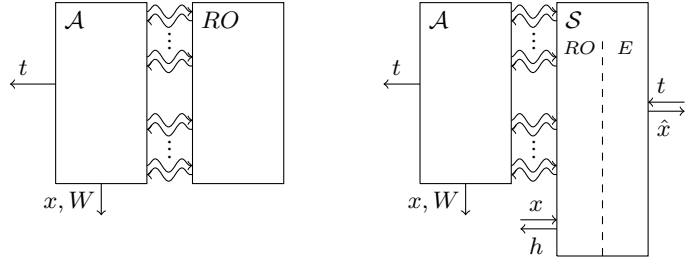
\begin{figure}[h]
$$
\begin{tikzpicture}[scale=0.12, baseline=0]
\draw (0,0) -- (0,20) node[anchor=north west] {$\cal A$} -- (10,20)  -- (10,0) -- (0,0) ;
\draw [->,snake it] (10,19) -- (15,19);
\draw [<-,snake it] (10,18) -- (15,18);
\node at (12.5,16.5)  {$\cdot$};
\node at (12.5,15.7)  {$\cdot$};
\node at (12.5,14.9)  {$\cdot$};
\draw [->,snake it] (10,14) -- (15,14);
\draw [<-,snake it] (10,13) -- (15,13);
\draw [->,snake it] (10,7) -- (15,7);
\draw [<-,snake it] (10,6) -- (15,6);
\node at (12.5,4.5)  {$\cdot$};
\node at (12.5,3.7)  {$\cdot$};
\node at (12.5,2.9)  {$\cdot$};
\draw [->,snake it] (10,2) -- (15,2);
\draw [<-,snake it] (10,1) -- (15,1);
\draw [<-] (-5,11) node[anchor=south west] {$t$} -- (0,11);
\draw [->] (5,0)  -- (5,-3.5) ;
\node at (1.5,-2) {$x,W$};
\draw (15,0) -- (15,20) node[anchor=north west] {$RO$} -- (25,20)  -- (25,0) -- (15,0) ;
\end{tikzpicture}
\qquad\qquad
\begin{tikzpicture}[scale=0.12, baseline=0]
\draw (0,0) -- (0,20) node[anchor=north west] {$\cal A$} -- (10,20)  -- (10,0) -- (0,0) ;
\draw [->,snake it] (10,19) -- (15,19);
\draw [<-,snake it] (10,18) -- (15,18);
\node at (12.5,16.5)  {$\cdot$};
\node at (12.5,15.7)  {$\cdot$};
\node at (12.5,14.9)  {$\cdot$};
\draw [->,snake it] (10,14) -- (15,14);
\draw [<-,snake it] (10,13) -- (15,13);
\draw [->,snake it] (10,7) -- (15,7);
\draw [<-,snake it] (10,6) -- (15,6);
\node at (12.5,4.5)  {$\cdot$};
\node at (12.5,3.7)  {$\cdot$};
\node at (12.5,2.9)  {$\cdot$};
\draw [->,snake it] (10,2) -- (15,2);
\draw [<-,snake it] (10,1) -- (15,1);
\draw [<-] (-4,11) node[anchor=south west] {$t$} -- (0,11);
\draw [->] (5,0)  -- (5,-3.5) ;
\node at (1.5,-2) {$x,W$};
\draw (15,-8) -- (15,20) node[anchor=north west] {$\cal S$} -- (25,20)  -- (25,-8) -- (15,-8) ;
\draw [dashed] (20,-8) -- (20,16);
\node at (17.5,15)  {\scriptsize $RO$};
\node at (22.5,15)  {\scriptsize $E$};
\draw [<-] (25,9) node[anchor=south west] {$t$} -- (29,9);
\draw [->] (25,8) node[anchor=north west] {$\hat x$} -- (29,8);
\draw [->] (11,-4) node[anchor=south west] {$x$} -- (15,-4);
\draw [<-] (11,-5) node[anchor=north west] {$h$} -- (15,-5);
\end{tikzpicture}
\vspace{-2ex}
$$
\caption{\submission{The original execution of}{} ${\cal A}^{RO}$ (left), and the experiment $G_{\mathcal S}^{\mathcal A}$ with $RO$ simulated by $\cal S$ (right). 
}\label{fig:AandS}
\end{figure}

Informally, we want that $\cal A$ does not notice any difference when $RO$ is replaced by $\SRO$, and that $\hat x = x$ whenever $f\bigl(x,h\bigr) = t$, while $\hat x = \emptyset$ implies that $\cal A$ will fail to output $x$ with $f\bigl(x,h\bigr) = t$. This situation is captured by the following statement. 

	\begin{corollary}
		\label{cor:toy}
The extractable RO-simulator $\cal S$ is such that the following holds. 
	For any ${\cal A}$ that outputs $t$ after $q_1$ queries and $x \in \cal X$ and $W$ after an additional $q_2$ queries, \submission{}{setting  $q = q_1+q_2$,} it holds that 
\begin{align*}
\delta\bigl([t,x,RO(x),W]_{{\cal A}^{RO}}, [t,x,h,W]_{G_{\mathcal S}^{\mathcal A}} \bigr) \leq  8 (q_2+1)  \sqrt{2 \Gamma/2^n} \qquad    &&\text{and} \\[1ex]
\Pr_{G_{\mathcal S}^{\mathcal A}}\bigl[x \neq \hat x \wedge f(x,h) = t \bigr]  \leq 8 (q_2+1) \sqrt{2 \Gamma/2^n} + \frac{40 e^2 (q+2)^3 \Gamma'(f)+2}{2^n} \, ,
\end{align*}
\submission{where $q = q_1+q_2$.}{}  
	\end{corollary}
	
	\begin{proof}
	The first claim follows from the fact that the trace distance vanishes when $\SE(t)$ is performed at the very end, after the $\SRO(x)$-query, in combination with the \mbox{(almost-)}commutativity of the two interfaces (Theorem~\ref{thm:MainFeatures}, 2.a to 2.c). Similarly, the second claim follows from Prop.~\ref{prop:HardCollision} when considering the $\SE(t)$ query to be performed at the very end, in combination with the (almost-)commutativity of the interfaces again. \qed
	\end{proof}

The statements above extend easily to {\em multi}-round algorithms ${\cal A}^{RO}$ that output $t_1,\ldots,t_\ell$ in (possibly) different rounds,  and $x_1,\ldots,x_\ell \in \cal X$ and some (possibly quantum) output $W$ at the end of the run.  We then extend the definition of $G_{\mathcal S}^{\mathcal A}$ in the obvious way: $\SE$ is queried on each output $t_i$ to produce $\hat x_i$, and at the end of the run $\SRO$ is queried on each of the final outputs $x_1,\ldots,x_\ell$ of $\mathcal A$ to obtain ${\bf h} = (h_1,\ldots,h_\ell)  \in {\cal Y}^\ell$.  As a minor extension, we allow some of the $x_i$ to be $\bot$, i.e., ${\cal A}^{RO}$ may decide to not output certain $x_i$'s; the $\SRO$ query on $x_i$ is then not done and $h_i$ is set to $\bot$ instead, and we declare that $RO(\bot) = \bot$ and $f(\bot,h_i) \neq t_i$. To allow for a compact notation, we write $RO({\bf x}) = (RO(x_1),\ldots,RO(x_\ell))$ for ${\bf x} = (x_1,\ldots,x_\ell)$. 

\begin{corollary}\label{cor:multi-round}
The extractable RO-simulator $\cal S$ is such that the following holds. 
For any ${\cal A}$ that makes $q$ queries in total, it holds that 
\begin{align*}
\delta\bigl([{\bf t},{\bf x},RO({\bf x}),W]_{{\cal A}^{RO}}, [{\bf t},{\bf x},{\bf h},W]_{G_{\mathcal S}^{\mathcal A}} \bigr) \leq 8  \ell(q+\ell) \sqrt{2\Gamma/2^n}      \qquad    \text{and}     \qquad \\[1ex]
\Pr_{G_{\mathcal S}^{\mathcal A}}  \bigl[ \exists \, i: x_i \neq \hat x_i \wedge f(x_i,h_i) = t_i \bigr]   
\leq 8 \ell (q+1) \sqrt{2 \Gamma/2^n} + \frac{40 e^2 (q+\ell+1)^3 \Gamma'(f)+2}{2^n}  .  
\end{align*}
	\end{corollary}

\submission{
	\begin{proof}
	The first claim follows from the fact that the trace distance vanishes when the $\SE(t_i)$-queries are performed at the very end, after all $\SRO(x_i)$-queries, in combination with the \mbox{(almost-)} commutativity of the interfaces. Similarly, the second claim follows from (the more general second part of) Prop.~\ref{prop:HardCollision} when considering the $\SE(t_i)$-queries to be performed at the very end, in combination with the \mbox{(almost-)}commutativity of the interfaces again. \qed
	\end{proof}
}{}

	\section{\submission{Application I: }{}Extractability of Commit-And-Open \sigps}\label{sec:CandOSigps}
	
	\oursubsection{Commit-and-Open \sigps}

	We assume the reader to be familiar with the concept of an interactive proof \submission{system }{}for a language $\mathcal L$ or a relation $R$, and specifically with the notion of a \emph{\sigp}.  \submission{We briefly discuss here the following special class of \sigps.}{} 
	
	Here, we consider the notion of a {\emph{commit-and-open} \sigp, which is as follows. The prover begins by sending commitments  $a_1,...,a_\ell$ to the prover, computed as $a_i=H(x_i)$ for $x_1,...,x_\ell \in \cal X$, where $H: {\cal X} \to \{0,1\}^n$ is a hash function\submission{, and where we assume for concreteness that $\cal X$ consists of bitstrings of bounded size}{}. Here, $x_i$ can either be the actual message $m_i$ to be committed, or $m_i$ concatenated with randomness. The verifier answers by sending a challenge $c$, which is a subset $c \subseteq [\ell]=\{1,...,\ell\}$, picked uniformly at random from a challenge set $ C\subseteq  2^{[\ell]}$, upon which the prover sends the response $z=(x_i)_{i\in c}$. Finally, the verifier checks whether $H(x_i)=a_i$ for every  $i\in c$, computes an additional verification predicate $V(c,z)$ and outputs 1 if both check out, 0 otherwise. Such (usually zero-knowledge) protocols have been known since the concept of zero-knowledge proofs was developed  \cite{BCC88,GMW91}. 
	
	Commit-and-open \sigps are (classically) extractable in a straight-forward manner as soon as a witness can be computed from sufficiently many of the $x_i$'s:  rewind the prover a few times until it has opened every commitment $a_i$ at least once.%
\footnote{Naturally, we can assume $[\ell]=\bigcup_{c\in  C}c$}
	There is, however, an alternative (classical) \emph{online} extractor if the hash function $H$ is modelled as a random oracle: simply look at the query transcript of the prover to find preimages of the commitments $a_1,...,a_\ell$. 
	As the challenge is chosen independently, the extractability and collision resistance of the commitments implies that for a prover with a high success probability, the $\ell$ extractions succeed simultaneously with good probability.  This is roughly how the proof of online extractability of the ZK proof system for graph 3-coloring by Goldreich, Micali and Wigderson \cite{GMW91}, instantiated with random-oracle based commitments, works that was announced in \cite{Pass03} and shown in \cite{Pass04} (Prop.~5).
	%
	%

Equipped with our extractable RO-simulator $\cal S$, we can mimmic the above in the quantum setting. Indeed, the only change is that the look-ups in the transcript are replaced with the additional interface of the simulator $\cal S$. Cor.~\ref{cor:multi-round} can then be used to prove the success of extraction using essentially the same extractor as in the classical case.
	

	\oursubsection{Notions of Special Soundness}
	
	The property that allows such an extraction is most conveniently expressed in terms of  special soundness and its variants. Because there are, next to special and $k$-soundness, a number of additional variants in the literature (e.g.\ in the context of Picnic2/Picnic3 \cite{KZ20} or MQDSS \cite{MQDSS}), we begin by formulating a generalized notion of special soundness that captures in a broad sense that a witness can be computed from correct responses to ``{\em sufficiently many\,}'' challenges.%
	\footnote{Using the language from secret sharing, we consider an arbitrary access structure $\frak{S}$, while the $k$-soundness case corresponds to a threshold access structure. }
While the notions introduced below can be formulated for arbitrary public-coin interactive proof systems, we present them here tailored to commit-and-open \sigps. 

In the remainder, $\Pi$ is thus assumed to be an arbitrary commit-and-open \sigp for a relation $R$ with associated language $\mathcal L$, and $C$ is the challenge space of $\Pi$. Furthermore, we consider a non-empty, monotone increasing set $\frak{S}$ of subsets $S \subseteq C$, i.e., such that $S \in \frak{S}\, \wedge \,S \subseteq S' \: \Rightarrow \: S' \in \frak{S}$, and we let $\frak{S}_{\min} := \{S \in \frak{S} \,|\, S_\circ \subsetneq S \Rightarrow S_\circ \not\in \frak{S} \}$ consist of the minimal sets in $\frak{S}$. 
	
	\begin{definition} \sloppy $\Pi$ is called $\frak{S}$\emph{-sound} if there exists an efficient algorithm $\mathcal E_\frak{S}(I,x_1,\ldots,x_\ell,S)$ that takes as input an instance $I \in \mathcal L$, strings $x_1,\ldots,x_\ell \in \cal X$ and a set $S \in \frak{S}_{\min}$, and outputs a witness for $I$ whenever $V(c,(x_i)_{i\in c})=1$ for all $c \in S$, and outputs $\bot$ otherwise.%
		\footnote{The restriction for $S$ to be in $\frak{S}_{\min}$, rather than in $\frak{S}$, is only to avoid an exponentially sized input\submission{ while asking $\mathcal E_\frak{S}$ to be efficient}{}. When $C$ is constant in size, we may admit any $S \in \frak{S}$.  }
	\end{definition}

Note that there is no correctness requirement on the $x_i$'s with $i \not\in \bigcup_{c \in S} c$; thus, those $x_i$'s may just as well be set to be empty strings. 
	

This property generalizes $k$-soundness, which is recovered for $\frak{S} = \frak{T}_k := \{S \subseteq C \,|\, |S| \geq k\}$, but it also captures more general notions. For instance, the $r$-fold parallel repetition of a $k$-sound protocol is not $k$-sound anymore, but it is $ \frak{T}_k^{\vee r}$-sound with $ \frak{T}_k^{\vee r}$ consisting of those subsets of challenge-sequences $(c_1,\ldots,c_r) \in C^r$ for which the restriction to at least one of the positions is a set in $ \frak{T}_k$. 
This obviously generalizes to the parallel repetition of an arbitrary $ \frak{S}$-sound protocol, with the parallel repetition then being $\frak{S}^{\vee r}$-sound with
$$
\frak{S}^{\vee r} := \{ S \subseteq C^r \,|\, \exists \, i: S_i \in \frak{S} \} \, ,
$$
where $S_i := \{c\in C \,|\, \exists \, (c_1,...,c_r)\in S: c_i=c \}$ is the $i$-th {\em marginal} of $S$.  

For our result to apply, we need a strengthening of the above soundness condition where $\mathcal E_\frak{S}$ has to find the set $S$ himself. 
This is clearly the case for $\frak{S}$-sound protocols that have a {\em constant sized} challenge space $C$, but also for the parallel repetition of $\frak{S}$-sound protocols with a constant sized challenge space. Formally, we require the following strengthened notion of $ \frak{S}$-sound protocols. 
	\begin{definition}
		\sloppy $\Pi$ is called $\frak{S}$\emph{-sound}$^*$ if there exists an efficient algorithm $\mathcal E^*_\frak{S}(I,x_1,\ldots,x_\ell)$ that takes as input an instance $I \in \mathcal L$ and strings $x_1,\ldots,x_\ell \in \cal X$, and outputs a witness for $I$ whenever there exists $S \in \frak{S}$ with $V(c,(x_i)_{i\in c})=1$ for all $c \in S$, and outputs $\bot$ otherwise. 	
	\end{definition}

$ \frak{S}$-sound \sigps may\,---\,and often do\,---\,have the property that a dishonest prover can pick any set $\hat S = \{\hat c_1,\ldots,\hat c_{m}\} \not\in  \frak{S}$ of challenges $\hat c_i \in C$ and then prepare $\hat x_1,\ldots,\hat x_\ell$ in such a way that $V(c,(\hat x_i)_{i \in c}) = 1$ if  
$c \in \hat S$, i.e., after having committed to $\hat x_1,\ldots,\hat x_\ell$ the prover can successfully answer challenge $c$ if $c \in \hat S$. We call this a {\em trivial} attack. The following captures the largest success probability of such a trivial attack, maximized over the choice of~$\hat S$:
\begin{align} \label{eq:ptriv}
p^{\mathfrak S}_{triv} := \frac{1}{|C|} \max_{\hat S \not\in  \frak{S}} |\hat S| \, . 
\end{align}
When there is no danger of confusion, we omit the superscript $\mathfrak S$. Looking ahead, our result will show that for any prover that does better than the trivial attack by a non-negligible amount, online extraction is possible. For special sound \sigps, $p_{triv} = 1/|C|$, and for $k$-sound \sigps, $p_{triv} = (k-1)/|C|$.  Furthermore, our definition of $\mathfrak{S}$-soundness allows a straightforward parallel repetition lemma on the combinatorial level providing an expression for $p_{triv}$ of parallel-repeated \sigps\submission{}{(the proof is an easy computation)}. 

\begin{lemma}
	Let $\Pi$ be an $\mathfrak S$-sound \sigp. Then 
	$p^{\mathfrak S^{\vee r}}_{triv}=\left(p^{\mathfrak S}_{triv}\right)^r$.
\end{lemma}

\submission{
\begin{proof}
	To prove the lemma, we simplify
	\begin{align*}
		p^{\mathfrak S^{\vee r}}_{triv}=\frac{1}{|C|^r} \max_{\hat S \not\in  \frak{S}^{\vee r}} |\hat S|
		=\frac{1}{|C|^r} \max_{\substack{\hat S \subset C^r: \\\forall i: \hat S_i\not\in\mathfrak S}} |\hat S|
		=\frac{1}{|C|^r}\left(\max_{\hat S \not\in\mathfrak S} |\hat S|\right)^r
	=\left(p^{\mathfrak S}_{triv}\right)^r.
	\end{align*}
\vspace{-.8cm}{}\\
\qed
\end{proof}
}{}


	

	\oursubsection{Online Extractability in the QROM}

	We are now ready to define our extractor and prove that it succeeds. 
	Equipped with the results from the previous section, the intuition is very simple. Given a (possibly dishonest) prover $\cal P$, running the considered \sigp in the QROM, we use the simulator $\cal S$ to answer $\cal P$'s queries to the random oracle but also to extract the commitments $a_1, \ldots, a_\ell$, and if the extracted $\hat x_1,\ldots,\hat x_\ell$ satisfy the verification predicate $V$ for sufficiently many challenges, we can compute a witness by applying ${\cal E}^*_\frak{S}$. 
	
	The following relates the success probability of this extraction procedure to the success probability of the (possibly dishonest) prover. 
	
	
	
	\begin{theorem}\label{thm:CnO-extract}
		Let $\Pi$ be an $\frak{S}$-sound$^*$ commit-and-open \sigp where the first message consists of $\ell$ commitments. Then it admits an online extractor $\mathcal E$ in the QROM that succeeds  with probability
%
\begin{align*}
	\Pr[\mathcal E\ \mathrm{succeeds}] \geq \frac{1}{1-p_{triv}} \bigl( \Pr[\mathcal P^{RO}\mathrm{\ succeeds}] - p_{triv} - \varepsilon \bigr)        &&       \text{where} \\[1ex]
%
				\varepsilon = 8 \sqrt{2} \, \ell(2q+\ell+1) /\sqrt{2^n} + \frac{40 e^2 (q+\ell+1)^3 \Gamma'(f)+2}{2^n} \, 
\end{align*}
                and $p_{triv}$ is defined in Eq.~\eqref{eq:ptriv}.
                 For $q\ge \ell+1$, the bound simplifies to
	\begin{equation*}
			\varepsilon\le 34 \ell q /\sqrt{2^{n}} +2365 q^3 /2^{n} \, .
	\end{equation*}
		Furthermore, the running time of $\cal E$ is bounded as 
		$
		T_{\mathcal E}=T_{{\cal P}_1}+T_{\mathcal E^*_\frak{S}}+O(q_1^2),
		$ 
		where $T_{{\cal P}_1}$ and $T_{\mathcal E^*_\frak{S}}$ are the respective runtimes of ${\cal P}_1$ and $\mathcal E^*_\frak{S}$.
	\end{theorem}
Recall that $p_{triv} = (k-1)/|C|$ for $k$-soundness, giving a corresponding bound. 

%
	
	\begin{proof}
		We begin by describing the extractor $\mathcal E$. \submission{In a first step}{First}, 
		using $\SRO$ to answer $\cal P$'s queries, $\cal E$ runs the prover $\cal P$ until it announces $a_1,\ldots,a_\ell$, and then it uses $\SE$ to extract $\hat x_1,...,\hat x_\ell$. I.e., 
		$\mathcal E$ acts as $\cal S$ in Cor. \ref{cor:multi-round} for the function $f(x,h) = h$ and runs the game $G_{\cal S}^{\cal P}$ to the point where $\SE$ outputs $\hat x_1,...,\hat x_\ell$ on input $a_1,\ldots,a_\ell$. As a matter of fact, for the purpose of the analysis, we assume that $G_{\cal S}^{\cal P}$ is run until the end, with the challenge $c$ chosen uniformly at random, and where ${\cal P}$ then outputs $x_i$ for all $i \in c$ (and $\bot$ for $i \not\in c$) at the end of $G_{\cal S}^{\cal P}$; we also declare that $\cal P$ additionally outputs $c$ and $a_1,\ldots,a_\ell$ at the end.  
Then, upon having obtained $\hat x_1,...,\hat x_\ell$, 
		the extractor $\cal E$ runs ${\cal E}^*_\frak{S}$ on $\hat x_1,...,\hat x_\ell$ to try to compute a witness. By definition, this succeeds if $\hat S := \{ \hat c \in C \, |\, V(\hat c,(\hat x_i)_{i\in \hat c})=1 \}$ is in $\frak{S}$. 
		
		It remains to relate the success probability of $\mathcal E$ to that of the prover $\mathcal P^{RO}$. 	
	By the first statement of Cor. \ref{cor:multi-round}, 
writing ${\bf x}_c = (x_i)_{i\in c}$, $RO({\bf x}_c) = (RO(x_i))_{i\in c}$, ${\bf a}_c = (a_i)_{i\in c}$, etc., 
we have 
	\begin{equation}\label{eq:bound0}
	\begin{split}
	\Pr[\mathcal P^{RO}\mathrm{\ succeeds}] &= \Pr_{\mathcal P^{RO}}[V(c,{\bf x}_c)=1\wedge RO({\bf x}_c)={\bf a}_c] \\
	&\leq \Pr_{G_{\cal S}^{\cal P}}[V(c,{\bf x}_c)=1\wedge {\bf h}_c={\bf a}_c] + \delta_1
	\end{split}
	\end{equation}
with $\delta_1=8 \sqrt{2} \, \ell(q+\ell) /\sqrt{2^n}$. Omitting the subscript $G_{\cal S}^{\cal P}$ now, 
\begin{equation}\label{eq:bound1}
\begin{split}
\Pr[V(c&,{\bf x}_c) =1\wedge {\bf h}_c = {\bf a}_c]  \\
	\leq&\Pr[V(c,{\bf x}_c)=1\wedge {\bf h}_c = {\bf a}_c \wedge {\bf x}_c={\bf \hat x}_c]
	+\Pr[{\bf h}_c = {\bf a}_c \wedge {\bf x}_c \neq {\bf \hat x}_c] \\ 
	\leq&\Pr[V(c,{\bf \hat x}_c)=1]  +\Pr[\exists \, j \in c: x_j\neq \hat x_j\wedge  h_j = a_j]   \\
	\leq&\Pr[V(c,{\bf \hat x}_c)=1]  + \delta_2   
	\end{split}
\end{equation}
with $\delta_2= 
8  \sqrt{2} \, \ell (q+1) /\sqrt{2^n}+ \frac{40 e^2 (q+\ell+1)^3 \Gamma'(f)+2}{2^n}$, where the last inequality is by the second statement of Cor.~\ref{cor:multi-round}, noting that, by choice of $f$, the event $h_j = a_j$ is equal to $f(x_j, h_j)=a_j$. 
		Recalling the definition of $\hat S$, 
	\begin{align}
	\Pr[V(c,{\bf \hat x}_c)=1] &= \Pr[c \in \hat S] \leq\Pr[\hat S \in \frak{S}] + \Pr[c \in \hat S \,|\, \hat S \not\in \frak{S}] \Pr[ \hat S \not\in \frak{S}] \label{eq:bound2}  \\
    	&\leq \Pr[\mathcal E\mathrm{\ succeeds}] + p_{triv} (1- \Pr[\mathcal E\mathrm{\ succeeds}] ) \nonumber
	\end{align}
where the final inequality exploits that $c$ is chosen at random and independent of $\hat x_1,\ldots,\hat x_\ell$, and thus is independent of the event $\hat S \not\in \frak{S}$. 
Combining (\ref{eq:bound0}), (\ref{eq:bound1}) and (\ref{eq:bound2}), we obtain
$$
\Pr[\mathcal P^{RO}\mathrm{\ succeeds}] \leq \Pr[\mathcal E\mathrm{\ succeeds}] + p_{triv} (1- \Pr[\mathcal E\mathrm{\ succeeds}] ) + \delta_1 + \delta_2
$$	
	and solving for  $\Pr[\mathcal E\mathrm{\ succeeds}]$ gives the claimed bound. 
\qed
	\end{proof}

\oursubsection{Tightness}	
The bound given by Theorem \ref{thm:CnO-extract} is tight 
	in the sense that the extraction success probability is proportional to the advantage of a malicious prover over the trivial success probability, up to a negligible additive error term.
	On top, the additive error term is asymptotically tight: $\varepsilon$ remains negligible in $n$ for $q = 2^{\alpha n}$ with any $\alpha < \frac13$, while with $q = 2^{n/3}$ queries a collision in the hash function can be found with constant success probability \cite{BHT98
,zhandry2015note}
, breaking the binding property of the commitment scheme upon which typical soundness proofs for commit-and-open \sigps rely. 
	
	
It is even not too hard to find relevant examples \submission{of commit-and-open \sigps }{}where a collision-finding attack not only invalidates the soundness proof but leads to an actual attack\submission{ against extractability}{}. Consider e.g.\ the \sigp ZKBoo that underlies the signature scheme Picnic. Here, the prover commits to three messages $m_1,m_2,m_3$ as $a_i = H(m_i,r_i)$ for random strings $r_1,r_2,r_3$, and where the $m_i$'s are the respective views of the three parties in an ``in-the-head'' execution of a 3-party-computation protocol. The challenge space is $C=\{\{1,2\}, \{1,3\}, \{2,3\}\}$, which means that the prover is then asked to open two out of the three commitments. 
	Now consider the following attack. 
	The attacker can easily find pairs $(m_1, m_2)$, $(m_1', m_3)$ and $(m_2', m_3')$, so that each pair consists of two mutually consistent views of the considered 3-party-computation protocol. 
	Now the only thing the attacker has to do is to find three collisions in the hash function of the form $a_i=H(m_i,r_i)=H(m'_i,r_i')$, $i=1,2,3$. This can be done using e.g. the BHT algorithm \cite{BHT98} if $r_i$ are sufficiently long. 
	The attacker now sends $(a_1,a_2,a_3)$, receives a challenge and responds with the appropriate preimages of the two commitments indicated by the challenge.

	\oursubsection{Application to Fiat Shamir Signatures}
	
	In \supmat\ (Sect.~\ref{Appendix:ApplSig}) we discuss the impact on Fiat Shamir signatures, in particular on the round-3 signature candidate Picnic \cite{Chase2017} in the NIST standardization process for post-quantum cryptographic schemes
	. In short, one crucial part in the chain of arguments to prove security of Fiat Shamir signatures is to prove that the underlying \sigp is a proof of knowledge. For post-quantum security, so far this step relied on Unruh's rewinding lemma, which leads (after suitable generalization), to a $(2k+1)$-th root loss for a $k$-sound protocols. For commit-and-open \sigps, Theorem~\ref{thm:CnO-extract} can replace Unruhs rewinding lemma when working in the QROM, making this step in the chain of arguments tight up to unavoidable additive errors.
	
As an example, Theorem \ref{thm:CnO-extract} implies a sizeable improvement over the current best QROM security proof of Picnic2 \cite{Chase2017,KZ20,CD+20}. Indeed, Unruh's rewinding lemma implies a 6-th root loss for the variant of special soundness the underlying \sigp possesses \cite{DFMS19}, while Theorem~\ref{thm:CnO-extract} is tight.

	\section{\submission{Application II: }{}QROM-Security of Textbook Fujisaki-Okamoto}\label{subsecFO}

\newcommand{\KEM}{\text{\sc kem}}
\newcommand{\PKE}{\text{\sc pke}}
	
	\oursubsection{The Fujisaki-Okamoto Transformation}
	
	The Fujisaki-Okamoto (FO) transform~\cite{FO99} is a general method to turn any public-key encryption scheme secure against {\it chosen-plaintext attacks} (CPA) into a key-encapsulation mechanism (KEM) that is secure against {\it chosen-ciphertext attacks} (CCA). We can start either from a scheme with one-way security against CPA attacks (OW-CPA) or from one with indistinguishability against CPA attacks (IND-CPA), and in both cases obtain an IND-CCA secure KEM. We recall that a KEM establishes a shared key, which can then be used for symmetric encryption. 
	
	We include the (standard) formal definitions of a public-key encryption scheme and of a KEM in \supmat, Section \ref{AppendixFO}, and we recall the notions of $\delta$-correctness and $\gamma$-spreadness 
	there. In addition, we define a relaxed version of the latter property, \emph{weak} $\gamma$-spreadness (see Definition \ref{def:weak-gamma}), where the ciphertexts are only required to have high min-entropy when averaged over key generation.\footnote{This seems relevant e.g. for lattice-based schemes, where the ciphertext has little (or even no) entropy for certain very unlikely choices of the key (like being all $0$)}.
	The security games for OW-CPA security of a public-key encryption scheme and for IND-CCA security of a KEM are given in Fig.~\ref{fig:SecDef}. 
	

	\begingroup
	\makeatletter
	\def\ALG@special@indent{%
		\ifdim\ALG@thistlm=0pt\relax
		\hskip-\leftmargin
		\else
		\hskip\ALG@thistlm
		\fi
	}%
	\newcommand{\Indcca}{}%
	\newcommand{\Decaps}{}%

\def\hyp{{\hbox{-}}}

\begin{figure}
\begin{center} \makebox[\textwidth][c]{ 
\fbox{ 
	\begin{minipage}[t]{0.4\linewidth}

			\begin{algorithmic}[1]
			\item[]\noindent\ALG@special@indent\underline{\sf{\bf GAME} OW-CPA}
			\State $(pk,sk)\leftarrow \sf Gen$
			\State $m^* \overset{\$}{\leftarrow}{\cal M}$
			\State $c^* \leftarrow {\sf Enc}_{pk}(m^*)$
			\State $m'\leftarrow {\cal A}(pk,c^*)$
			\State \Return $m' == m^*$\label{alg:l1}
			\end{algorithmic}
	\end{minipage}%
	\hfill
	\begin{minipage}[t]{0.333\linewidth}
		\begin{algorithmic}[1]
			\setcounterref{ALG@line}{alg:l1}
			\item[]\noindent\ALG@special@indent\underline{\sf{\bf GAME} IND-CCA-KEM}
			\State $(pk,sk)\leftarrow \sf Gen$
			\State $b \overset{\$}{\leftarrow}\{0,1\}$
			\State $(K_0^*,c^*) \leftarrow {\sf Encaps}(pk)$
			\State $K_1^*\overset{\$}{\leftarrow}\cal K$
			\State $b'\leftarrow {\cal A^{\textsc{Decaps}}}(c^*,K_b^*)$
			\State \Return $b' == b$\label{alg:KEMl2}
		\end{algorithmic}
			
	\end{minipage}%
\hfill
\begin{minipage}[t]{0.267\linewidth}

	\begin{algorithmic}[1]
		\setcounterref{ALG@line}{alg:KEMl2}
		\item[]\noindent\ALG@special@indent\underline{\textsc{Decaps}$(c\neq c^*)$}
		\State $K:= {\sf Decaps}_{sk}(c)$
		\State \Return $K$
	\end{algorithmic}

	\end{minipage}%
	
	\hfill
}
}\end{center}\vspace{-.4cm}
\caption{Games for OW-CPA security of a PKE and IND-CCA security of a KEM. In the latter, $\cal A$ is not allowed to query $c^*$ to {\sc Decaps}.}\label{fig:SecDef}
\end{figure}

The formal specification of the FO transformation, mapping a public-key encryption scheme ${\sf PKE} = \sf ({Gen}, {Enc}, {Dec})$ and two suitable hash functions $H$ and $G$ (which will then be modeled as random oracles) into a key encapsulation mechanism ${\sf FO[PKE},H,G] = (\sf {Gen, Encaps, Decaps})$, is given in Fig.~\ref{fig:FO}.
	

\begin{figure}
	\centering
	\begin{center}\makebox[\textwidth][c]{\fbox{%
		\begin{minipage}[t]{0.28\linewidth}

			\begin{algorithmic}[1]
				\item[]\noindent\ALG@special@indent\underline{${\sf Gen}$}
				\State 	$(sk,pk)\leftarrow {\sf Gen}$
				\State \Return $(sk,pk)$
				\label{alg:FO1}
			\end{algorithmic}
		\end{minipage}%
		\begin{minipage}[t]{0.3\linewidth}

			\begin{algorithmic}[1]
				\item[]\noindent\ALG@special@indent\underline{${\sf Encaps}(pk)$}
				\setcounterref{ALG@line}{alg:FO1}
				\State $m\overset{\$}{\leftarrow}\cal M$
				\State $c\leftarrow {\sf Enc}_{pk}(m;H(m))$
				\State $K:= G(m)$
				\State \Return $(K,c)$
				\label{alg:FO2}
			\end{algorithmic}
		\end{minipage}%
		\begin{minipage}[t]{0.44\linewidth}

			\begin{algorithmic}[1]
				\setcounterref{ALG@line}{alg:FO2}
				\item[]\noindent\ALG@special@indent\underline{${\sf Decaps}_{sk}(c)$}
				\State $m := {\sf Dec}_{sk}(c)$
				\State {\bf if} $m =\,\bot$ {\bf or} ${\sf Enc}_{pk}(m;H(m)) \neq c$\par\noindent\hskip 1em \Return $\bot$
				\State {\bf else} \Return $K:= G(m)$
			\end{algorithmic}
			
		\end{minipage}%
		\hfill
              }}\end{center}\vspace{-.4cm}
              \caption{The KEM ${\sf FO[PKE},H,G]$, obtained by applying the FO transformation~\cite{FO99} to $\sf PKE$. }\label{fig:FO}
\end{figure}

	\oursubsection{Post-Quantum Security of FO in the QROM}
	
	Our main contribution here is the following security result for the FO transformation in the QROM. In contrast to most of the previous works on the topic, our result applies to the {\em standard} FO transformation, without any adjustments. Next to being CPA secure, we require the underlying public-key encryption scheme to be so that ciphertexts have a lower-bounded amount of min-entropy (resulting from the encryption randomness), captured by the aforementioned spreadness property. This seems unavoidable for the FO transformation with explicit rejection and without any adjustment, like an additional key confirmation hash (as e.g. in \cite{TU16}). 

	\begin{theorem}\label{thm:FO}
		Let {\sf PKE} be a $\delta$-correct public-key encryption scheme satisfying weak $\gamma$-spreadness. Let $\cal A$ be any {\sf IND-CCA} adversary against ${\sf FO[PKE},H,G]$, making $q_D \geq 1$ queries to the decapsulation oracle {\sc Decaps} and $q_H$ and $q_G$ queries to $H: {\cal M}\rightarrow {\cal R}$ and $G : {\cal M}\rightarrow {\cal K}$, respectively, where $H$ and $G$ are modeled as random oracles. Let $q:= q_H + q_G + 2q_D$. Then, there exists a {\sf OW-CPA} adversary $\cal B$ against {\sf PKE} with
		\begin{align*}
			{\sf ADV[{\cal A}]^{\sf IND\text{-}CCA}_\KEM} 
			\leq& \; 2q\sqrt{{\sf ADV^{\sf OW\text{-}CPA}_\PKE[\cal B]}} +24q^2\sqrt{ \delta} + 24q\sqrt{q q_D}\cdot 2^{-\gamma/4} \, .
		\end{align*}
	Furthermore, $\cal B$ has a running time $T_{\cal B} \leq T_{\cal A} + O\bigl(q_H\cdot q_D\cdot \mathrm{Time}[{\sf Enc}] + q^2\bigr).$
	\end{theorem}
	
We start with a proof outline, which is somewhat simplified in that it treats ${\sf FO[PKE},H,G]$ as an encryption scheme rather than as a KEM. We will transform the adversary $\cal A$ of the {\sf IND-CCA} game into a {\sf OW-CPA} adversary against the $\sf PKE$ in a number of steps. 
	There are two main challenges to overcome. (1) We need to switch from the {\em deterministic} challenge ciphertext $c^* = {\sf Enc}_{pk}(m^*;H(m^*))$ that $\cal A$ attacks to a {\em randomized} challenge ciphertext $c^* = {\sf Enc}_{pk}(m^*;r^*)$ that $\cal B$ is then supposed to attack. We do this switch by re-programming $H(m^*)$ to a random value right after the computation of $c^*$, which  is equivalent to keeping $H$ but choosing a random $r^*$ for computing  $c^*$. For reasons that we explain later, we do this switch from $H$ to its re-programmed variant, denoted $H^\diamond$, in two steps, where the first step (from {\bf Game 0} to {\bf 1}) will be ``for free'', and the second step (from {\bf Game 1} to {\bf 2}) is argued using the O2H lemma (\cite{OriginalO2H}, we use the version given in \cite{AHU19}, Theorem~3). (2) We need to answer decryption queries without knowing the secret key. At this point our extractable RO-simulator steps in. We replace $H^\diamond$, modelled as a random oracle, by $\cal S$, and we use its extraction interface to extract $m$ from any correctly formed encryption $c ={\sf Enc}_{pk}(m;H^\diamond(m))$ and to identify incorrect ciphertexts.  
	
One subtle issue in the argument above is the following. The O2H lemma ensures that we can find $m^*$ by measuring one of the queries to the random oracle. However, given that also the decryption oracle makes queries to the random oracle (for performing the re-encryption check), it could be the case that one of those decryption queries is the one selected by the O2H extractor. This situation is problematic since, once we switch to $\cal S$ to deal with the decryption queries, some of these queries will be dropped (namely when $\SE(c) = \emptyset$).  
This is problematic because, per-se, we cannot exclude that this is the one query that will give us $m^*$.
We avoid this problem by our two-step approach for switching from $H$ to $H^\diamond$, which ensures that the only ciphertext $c$ that would bring us in the above unfortunate situation is the actual (randomized) {\em challenge ciphertext} $c^* ={\sf Enc}_{pk}(m^*;r^*)$, which is \submission{not submitted}{forbidden} by the specification of the security game. 

\begin{figure*}
	\makebox[\textwidth][c]{\fbox{%
			\hfill
			\begin{minipage}[t]{\switch{0.46}{0.525}\linewidth}
				\begin{algorithmic}[1]
					\item[]\noindent\ALG@special@indent\underline{\textsc{Game Setup} $G_0\hyp G_8$}\vspace{2pt}
					\State $(pk,sk)\leftarrow \sf Gen$\hfill$\sslash G_0\hyp G_7$
					\State $(b,m^*) \overset{\$}{\leftarrow}\{0,1\}\times {\cal M}$\hfill$\sslash G_0\hyp G_7$
					\State $c^*:={\sf Enc}_{pk}(m^*;H(m^*))$\hfill$\sslash G_0\hyp G_7$
					\State \textbf{input}($pk,c^*={{\sf Enc}_{pk}(m^*)}$)\hfill$\sslash G_8$
					\State {\color{red}$c^\diamond:={\sf Enc}_{pk}(m^*;H^\diamond(m^*))$\hfill$\sslash G_0\hyp G_6$}
					\State $K_0^*:= G(m^*)$\hfill$\sslash G_0\hyp G_2$
					\State $K_1^*\overset{\$}{\leftarrow}\cal K$
					\State {\color{blue}$j\overset{\$}{\leftarrow}J_{\cal A}\cup J_{D({\color{red} c^\diamond})}$\hfill$\sslash G_3\hyp G_6$}
					\State {\color{blue}$j\overset{\$}{\leftarrow}J$\hfill$\sslash G_7\hyp G_8$}
					\item[]
					\item[]\noindent\ALG@special@indent\underline{\textsc{Main Phase} $G_0\hyp G_2$}\vspace{2pt}
					\State $b'\leftarrow {{\cal A}^{\textsc{Decaps},H,G}}(c^*,K_b^*)$\hfill$\sslash G_0\hyp G_1$
					\State $b'\leftarrow {{\cal A}^{\textsc{Decaps},H^\diamond,G^\diamond}}(c^*,K_b^*)$\hfill$\sslash G_2$
					\State \Return $b' == b$	
					\item[]
					\item[]\noindent\ALG@special@indent\underline{\textsc{Main Phase} $G_3\hyp G_8$}\vspace{2pt}
					\State {\color{blue}$m'\leftarrow {{\cal M\!A}_{j}^{\textsc{Decaps},H^\diamond,G^\diamond}}(c^*,K_1^*)$\hfill$\sslash G_3$}
					\State {\color{blue}$m'\leftarrow {{\cal M\!A}_{j}^{\textsc{Decaps},\SRO,G^\diamond}}(c^*,K_1^*)$\hfill$\sslash G_4\hyp G_5$}
					\State {\color{blue}$m'\leftarrow {{\cal E\!A}_{j}^{\textsc{Decaps},\SRO,G^\diamond}}(c^*,K_1^*)$\hfill$\sslash G_6\hyp G_8$}
					\State {\color{violet}\textbf{while }$i \in I$\textbf{ do }\hfill$\sslash G_4$}\\
					{\color{violet}\hskip 1em $\hat{m}_i \leftarrow \SE(c_i)$\hfill$\sslash G_4$}
					\State \Return $m'$
					\label{alg:l2}
				\end{algorithmic}
				
			\end{minipage}%
			\hfill
			\begin{minipage}[t]{\switch{0.06}{0.01}\linewidth}
				\hfill~	
			\end{minipage}
			\algtext*{EndIf}
			\begin{minipage}[t]{\switch{0.46}{0.47}\linewidth}
				\begin{algorithmic}[1]
					\setcounterref{ALG@line}{alg:l2}
					\item[]\noindent\ALG@special@indent\underline{\textsc{Decaps}$(c\neq c^*)$ $G_0\hyp G_5$}\vspace{2pt}
					\State $m := {\sf Dec}_{sk}(c)$\hfill$\sslash G_0\hyp G_5$
					\State {\bf if} $m =\,\bot$ \Return $\bot$\hfill$\sslash G_0\hyp G_5$
					\State $h:= H(m), g:=G(m)$\hfill$\sslash G_0$
					\State {\color{red}{\bf if} $c=c^\diamond$\hfill$\sslash G_1$}
					\State {\color{red}\hskip 1.5 em $h:= H(m), g:=G(m)$\hfill$\sslash G_1$}
					\State {\color{red}\textbf{else}\hfill${\sslash G_1}$}
					\State {\color{red}\hskip 1.5 em $h:= H^\diamond(m), g:=G^\diamond(m)$\hfill$\sslash G_1$}
					\State $h:= H^\diamond(m), g:=G^\diamond(m)$\hfill$\sslash G_2\hyp G_3$
					\State $h:= \SRO(m), g:=G^\diamond(m)$\hfill$\sslash G_3\hyp G_5$
					\State {\bf if} ${\sf Enc}_{pk}(m;h) \neq c$ \hfill${\sslash G_0\hyp G_5}$
					\State \ \ \ \Return $\bot$\hfill$\sslash G_0\hyp G_5$
					\State \textbf{else} \Return $K:=g$\hfill${\sslash G_0\hyp G_5}$
					\State {\color{violet}$\hat{m} \leftarrow \SE(c)$\hfill$\sslash G_5$}
					\item[]
					\item[]\noindent\ALG@special@indent\underline{\textsc{Decaps}$(c\neq c^*)$ 
					$G_6\hyp G_8$}\vspace{2pt}
					\State $m := {\sf Dec}_{sk}(c)$\hfill$\sslash G_6\hyp G_7$
					\State \textbf{query} $\SRO(m)$ \hfill$\sslash G_6 \hyp G_7$
					\State {\color{violet}$\hat{m} \leftarrow \SE(c)$\hfill$\sslash G_6\hyp G_8$}
					\State {\bf if} $\hat{m} =\,\bot$ \Return $\bot$\hfill$\sslash G_6\hyp G_8$
					\State \bf {else} \Return $K:= G^\diamond(\hat{m})$\hfill$\sslash G_6\hyp G_8$
				\end{algorithmic}		
			\end{minipage}%
			\hfill
	}}\caption{{\bf Games 0}  to {\bf 8}. $H$ and $G$ are independent random oracles;  $H^\diamond$ and $G^\diamond$ coincide with $H$ and $G$, respectively, except that $H^\diamond(m^*)$ and $G^\diamond(m^*)$ are freshly chosen. 
		We consider the oracle queries to $H^\diamond$ (respectively to $\SRO$ later on) and to $G^\diamond$ to be labeled by indices $j \in J$, where 
		$J = J_{\cal A} \cup J_{D}$ decomposes this set into those queries made by $\cal A$ and those made by {\sc Decaps}, respectively, and $J_{D(c^\diamond)} \subseteq J_{D}$ consists of {\sc Decaps}' queries upon input $c^\diamond$. Similarly, we consider the queries to {\sc Decaps} to be indexed by $i \in I$, with $c_i$ then being the corresponding ciphertext. Since $\cal A$ is not allowed to query $c^*$ to {\sc Decaps}, we have $c_i\neq c^*$ $\forall \, i\in I$.
		For $j \in J$, ${\cal M\!A}_j^{\textsc{Decaps}}$ denotes the execution of ${\cal A}^{\textsc{Decaps}}$ up to the query indexed by $j$, and followed by measuring this query and outputting the result. ${\cal E\!A}_j^{\textsc{Decaps}}$ coincides with ${\cal M\!A}_j^{\textsc{Decaps}}$, except that if $j \in J_D$ then it outputs the corresponding $\hat m_i$ instead. The colors are meant to help the reader track (the use of) some variables and concepts that occur in different places across the code.}\label{fig:FOpseudocode}
	
\end{figure*}

	\begin{proof}[of Theorem \ref{thm:FO}]
			~ {\bf Games 0} to {\bf 8} below show how to turn ${\cal A}$ into ${\cal B}$ (see also Figure \ref{fig:FOpseudocode}).  We first analyze the sequence of hybrids for a fixed key pair $(sk,pk)$. Let therefore \smash{${\sf ADV}_{sk}{\sf[A]}^{\sf IND\text{-}CCA}_\KEM$} be $\mathsf A$'s advantage for key pair $(sk, pk)$. In addition, for a fixed pair $(sk, pk)$, let $\delta_{sk}$ be the maximum probability of a decryption error and $g_{sk}$ be the maximum probability of any ciphertext, so that $\mathbb E\bigl[\delta_{sk} \bigr] \le \delta$ and $\mathbb E\bigl[ g_{sk} \bigr] \le 2^{-\gamma}$, with the expectation over ${(sk,pk)\leftarrow\mathsf{Gen}}$ (we can assume without loss of generality that $pk$ is included in $sk$).

		\textbf{Game 0} is the {\sf IND-CCA} game for KEMs, except that we replace the random oracles $G$ and $H$ with a single random oracle $F$, by setting $H(x):= F(0||x)$ and $G(x):=F(1||x)$.\footnote{These assignments seem to suggest that ${\cal R} = {\cal K}$, which may not be the case. Indeed, we understand here that $F: {\cal M} \to \{0,1\}^n$ with $n$ large enough, and $F(0||x)$ and $F(1||x)$ are then cut down to the right size.  }
		When convenient, we still refer to $F(0\|\cdot )$ as $H$ and $F(1\|\cdot )$ as $G$. This change does not affect the view of the adversary nor the outcome of the game; therefore, 
		$$
		\Pr[ b = b' \text{ in {\bf Game 0}}] = \frac{1}{2} + {\sf ADV}_{sk}{\sf[A]}^{\sf IND\text{-}CCA}_\KEM.
		$$
		
		In \textbf{Game 1}, we introduce a new oracle $F^\diamond$ by setting $F^\diamond(0\|m^*) := r^\diamond$ and $F^\diamond(1\|m^*):= k^\diamond$ for uniformly random $r^\diamond\in {\cal R}$ and $k^\diamond\in \cal K$, while letting $F^\diamond(b\|m) := F(b\|m)$ for $m\neq m^*$ and $b\in\{0,1\}$. We note that while the {\em joint} behavior of $F^\diamond$ and $F$ depends on the choice of the challenge message $m^*$, each one individually is a purely random function, i.e., a random oracle. In line with $F$, we write $H^\diamond$ for $F^\diamond(0\|\cdot)$ and $G^\diamond$ for $F^\diamond(1\|\cdot)$ when convenient.

		Using these definitions, \textbf{Game 1} is obtained from \textbf{Game 0} via the following modifications. After $m^*$ and $c^*$ have been produced and before $\cal A$ is executed, we compute $c^\diamond :=  {\sf Enc}_{pk}(m^*;r^\diamond) =  {\sf Enc}_{pk}(m^*;H^\diamond(m^*))$, making a query to $H^\diamond$ to obtain $r^\diamond$. 
		Furthermore, for every decapsulation query by $\cal A$, we let {\sc Decaps} use $H^\diamond$ and $G^\diamond$ instead of $H$ and $G$ for checking correctness of the queried ciphertexts $c_i$ and for computing the key $K_i$, {\em except} when $c_i=c^\diamond$ (which we may assume to happen at most once), in which case {\sc Decaps} still uses $H$ and $G$. We claim that 
		$$
		\Pr[ b = b' \text{ in {\bf Game 1}}] = \Pr[ b = b' \text{ in {\bf Game 0}}] = \frac{1}{2} + {\sf ADV}_{sk}{\sf[A]}^{\sf IND\text{-}CCA}_\KEM \, .
		$$
		Indeed, for any decryption query $c_i$, we either have ${\sf Dec}_{sk}(c_i) =: m_i \neq m^*$ and  thus $F^\diamond(b\|m_i) = F(b\|m_i)$, or else $m_i = m^*$; in the latter case we then either have $c_i = c^\diamond$, where nothing changes by definition of the game, or else ${\sf Enc}_{pk}(m^*;H(m^*))  =c^*\neq c_i \neq c^\diamond = {\sf Enc}_{pk}(m^*;H^\diamond(m^*))$, and hence the re-encryption check fails and $K_i := \bot$ in either case, without querying $G$ or $G^\diamond$. 
		Therefore, the input-output behavior of $\sf Decaps$ is not affected. 
		
		\smallskip

	In \textbf{Game 2}, all oracle calls by $\sf Decaps$ (also for $c_i = c^\diamond$) and all calls by $\cal A$ are now to~$F^\diamond$. Only the challenge ciphertext $c^* = {\sf Enc}_{pk}(m^*;H(m^*))$ is still computed using $H$, and thus with randomness $r^* = H(m^*)$ that is random and independent of $m^*$ and $F^\diamond$. Hence, looking ahead, we can think of $c^*$ as the input to the {\sf OW-CPA} game that the to-be-constructed attacker $\cal B$ will attack. Similarly, $K^*_0 = G(m^*)$ is random and independent of $m^*$ and $F^\diamond$, exactly as $K^*_1$ is, which means that $\cal A$ can only win with probability~$\frac12$. 
	
	By the O2H lemma (\cite{AHU19}, Theorem~3), the difference between the respective probabilities of $\cal A$ in guessing $b$ in \textbf{Game 1} and \textbf{2} gives a lower bound on the success probability of a particular procedure to find an input on which $F$ and $F^\diamond$ differ, and thus to find $m^*$. Formally, 
			\begin{align*}
		2(q_H+q_G+2)&\sqrt{\Pr[\text{$m' = m^*$ in \textbf{Game 3}}]} \\
		&\hskip 2.5em\geq |\Pr[b'=b\text{ in \textbf{Game 1}}] - \Pr[b'=b\text{ in \textbf{Game 2}}]| \\
		&\hskip 2,5em= \frac{1}{2}+{\sf ADV}_{sk}{\sf[A]}^{\sf IND\text{-}CCA}_\KEM - \frac{1}{2}\\
		 &\hskip 2.5em= {\sf ADV}_{sk}{\sf[A]}^{\sf IND\text{-}CCA}_\KEM
		\end{align*}
	where \textbf{Game 3} is identical to \textbf{Game 2} above, except that we introduce and consider a new variable $m'$ (with the goal that $m' = m^*$), obtained as follows. Either one of the $q_H+q_G$ queries from $\cal A$ to $H^\diamond$ and $G^\diamond$ is measured, or one of the two respective queries from {\sc Decaps} to $H^\diamond$ and $G^\diamond$ upon a possible decryption query $c^\diamond$ is measured, and, in either case, $m'$ is set to be the corresponding measurement outcome. The choice of which of these $q_H+q_G+2$ queries to measure is done uniformly at random.%
	\footnote{If this choice instructs to measure {\sc Decaps}'s query to $H^\diamond$ or to $G^\diamond$ for the decryption query $c^\diamond$, but there is no decryption query $c_i = c^\diamond$, $m' := \bot$ is output instead.}

	We note that, since we are concerned with the measurement outcome $m'$ only, it is irrelevant whether the game stops right after the measurement, or it continues until $\cal A$ outputs $b'$. 
	Also, rather than actually measuring {\sc Decaps}' classical query to $H^\diamond$ or $G^\diamond$ upon decryption query $c_i = c^\diamond$ (if instructed to do so), we can equivalently set $m' := m_i = {\sf Dec}_{sk}(c^\diamond)$. 
		
		For \textbf{Game 4}, we consider the function $f : {\cal M} \times {\cal R}\rightarrow {\cal C}$, $(m,r) \mapsto {\sf Enc}_{pk}(m;r)$, and we replace the random oracle $H^\diamond$ with the extractable RO-simulator $\cal S$ from Theorem~\ref{thm:MainFeatures}. 
		Furthermore, {\em at the very end} of the game, we invoke the extractor interface $\SE$ to compute $\hat m_i := \SE(c_i)$ for each $c_i$ that $\sf A$ queried to {\sc Decaps} in the course of its run. By the first statement of Theorem \ref{thm:MainFeatures}, given that the $\SE$ queries take place only \textit{after} the run of $\cal A$, 
		$$
		\Pr[\text{$m' = m^*$ in \textbf{Game 4}}]	= \Pr[\text{$m' = m^*$ in \textbf{Game 3}}] \, .
		$$
		Furthermore, applying Prop.~\ref{prop:multiHardProperty} for $R' := \{(m,c) : {\sf Dec}_{sk}(c)\neq m\}$, we get that the event
			$$
			P^\dagger := \big[\, \forall i:  \hat{m}_i  = m_i \vee \hat{m}_i =\emptyset \big] 
			$$
			holds except with probability $\varepsilon_1 :=  128 (q_H+q_D)^2\Gamma_R/|{\cal R}|$ for $\Gamma_R$ as in Prop.~\ref{prop:multiHardProperty}, which here means that $\Gamma_R/|{\cal R}| = \delta_{sk}$.
			 Thus
		$$
		\Pr[m' = m^* \,\wedge\, P^\dagger\text{ in \textbf{Game 4}}] \geq \Pr[\text{$m' = m^*$ in \textbf{Game 4}}]-\varepsilon_1 \, .
		$$
		
		\smallskip
		
		In \textbf{Game 5}, we query $\SE(c_i)$ {\em at runtime}, that is, as part of the {\sc Decaps} procedure upon input~$c_i$, right after $\SRO(m)$ has been invoked as part of the re-encryption check (line 27 of Figure \ref{fig:FOpseudocode}). Since $\SRO(m)$ and $\SE(c_i)$ now constitute two subsequent classical queries, it follows from the contraposition of 4.b of Theorem \ref{thm:MainFeatures} that except with probability $2\cdot 2^{-n}$, $\hat{m}_i =\emptyset$ implies ${\sf Enc}_{pk}(m_i;\SRO(m_i)) \neq c_i$. Applying the union bound, we find that $P^\dagger$ implies 
		$$
		P := \big[ \,\forall i:  \hat{m}_i  =  m_i \,\vee\, ( \hat{m}_i = \emptyset \,\wedge\, {\sf Enc}_{pk}(m_i;\SRO(m_i))\neq c_i) \big]
		$$
		except with probability $q_D\cdot 2\cdot 2^{-n}$.
		Furthermore, By 2.c of that same Theorem \ref{thm:MainFeatures}, each swap of a $\SRO$ with a $\SE$ query affects the final probability by at most $8\sqrt{2\Gamma(f)/|{\cal R}|} = 8\sqrt{2 g_{sk}}$. 
		Thus
		$$ 
		\Pr[m' = m^* \,\wedge\, P \text{ in \textbf{Game 5}}] \geq \Pr[m' = m^* \,\wedge\, P^{\dagger} \text{ in \textbf{Game 4}}]-\varepsilon_2
		$$ 
		with $\varepsilon_2:=2q_D\cdot \left((q_H+q_D)\cdot 4\sqrt{2 g_{sk}}+ 2^{-n}\right)$.
		
		\smallskip
		
		In \textbf{Game 6}, {\sc Decaps} uses $\hat{m}_i $ instead of $m_i$ to compute $K_i$. That is, it sets $K_i := \bot$ if $\hat{m}_i =\emptyset$ and $K_i := G^\diamond(\hat{m}_i )$ otherwise. 
		Also, if instructed to output $m' := m_i$ where $c_i = c^\diamond$, then the output is set to $m' := \hat m_i$ instead. 
		In all cases, {\sc Decaps}  still queries $\SRO(m_i)$, so that the interaction pattern between {\sc Decaps} and $\SRO$ remains as in \textbf{Game~5}. 
		
		Here, we note that if the event $$P_i := \big[ \hat{m}_i  =  m_i \,\vee\, ( \hat{m}_i = \emptyset \,\wedge\, {\sf Enc}_{pk}(m_i;\SRO(m_i))\neq c_i) \big]$$ holds for a given $i$ then the above change will not affect {\sc Decaps}' response $K_i$, and thus also not the probability for $P_{i+1}$ to hold as well. Therefore, by induction, $\Pr[P \text{ in \textbf{Game 6}}] = \Pr[P \text{ in \textbf{Game 5}}]$, and since conditioned on the event $P$ the two games are identical, we have
		$$
		\Pr[m' = m^* \,\wedge\, P \text{ in \textbf{Game 6}}] = \Pr[m' = m^* \,\wedge\, P \text{ in \textbf{Game 5}}].
		$$

		In \textbf{Game 7}, instead of obtaining $m'$ by measuring a random query of $\cal A$ to either $\SRO$ or $G$, or outputting $\hat m_i$ with $c_i = c^\diamond$, here $m'$ is obtained by measuring a random query of $\cal A$ to either $\SRO$ or $G$, or outputting $\hat m_{i}$ for a {\em random}~$i \in \{1,\ldots,q_D\}$, where the former case is chosen with probability $(q_H+q_G)/(q_H+q_G+2q_D)$ and the latter with probability $2q_D/(q_H+q_G+2q_D)$. Since conditioned on the first case being chosen or the latter with $i = i_\diamond$,  \textbf{Game~7} coincides with  \textbf{Game~6}, we have
		$$
		\Pr[\text{$m' = m^*$ in \textbf{Game 7}}] \geq \frac{q_H+q_G+2}{q_H+q_G+2q_D}\cdot\Pr[\text{$m' = m^*$ in \textbf{Game 6}}] \, .
		$$
		
		In \textbf{Game 8}, we observe that the response to the query $\SRO(m^*)$, introduced in {\bf Game 1} in order to compute $c^\diamond$, and the responses to the queries that {\sc Decaps} makes to $\SRO$ on input $m_i$ do not affect the game anymore, and thus we can drop all these queries, or, equivalently, move them to the very end of the execution of the game. Invoking once again  2.c of Theorem~\ref{thm:MainFeatures}, we then get
		$$
		\Pr[\text{$m' = m^*$ in  \textbf{Game 8}}] \geq \Pr[\text{$m' = m^*$ in \textbf{Game 7}}] - \varepsilon_3 \, ,
		$$
		for $\varepsilon_3 = (q_D+1)\cdot q_H\cdot 8\sqrt{2 g_{sk}}$. 
		
		With these queries now dropped, we observe that \textbf{Game 8}  works without knowledge of the secret key $sk$, and thus constitutes a $\sf OW\text{-}CPA$ attacker $\cal B$ against $\sf PKE$, which takes as input a public key $pk$ and an encryption $c^*$ of a random message $m^* \in \cal M$, and outputs $m^*$ with the given probability, i.e,  ${\sf ADV}_{sk}{\sf[B]}^{\sf OW\text{-}CPA}_\PKE \geq \Pr[\text{$m' = m^*$ in  \textbf{Game 8}}]$. We note that the oracle $G^\diamond$ can be simulated using standard techniques. 
		 
		Backtracking all the above (in)equalities and setting $\varepsilon_{23} := \varepsilon_2 + \varepsilon_3$, $q_{HG}:= q_H+q_G$ etc. and $q:=q_H+q_G+2q_D$, we get the following bound:
		\begin{align*}
		{\sf ADV}_{sk}{\sf[{\cal A}]^{\sf IND\text{-}CCA}_\KEM} &\leq 2(q_{HG}+2) \sqrt{ \frac{q_{HG}+2q_D}{q_{HG}+2}\big( {\sf ADV}_{sk}{\sf[B]}^{\sf OW\text{-}CPA}_\PKE + \varepsilon_3\big) + \varepsilon_{1} + \varepsilon_{2}} \\
		&\leq 2(q_{HG}+2q_D) \sqrt{{\sf ADV}_{sk}{\sf[B]}^{\sf OW\text{-}CPA}_\PKE +  \varepsilon_{23}}+ 2(q_{HG}+2) \sqrt{ \varepsilon_{1}}  \\
		&\leq 2q \Big( \sqrt{{\sf ADV}_{sk}{\sf[B]}^{\sf OW\text{-}CPA}_\PKE}  + \sqrt{ \varepsilon_{23}} + \sqrt{ \varepsilon_{1}}\Big) \, .
		\end{align*}
		Additionally,
		\begin{align*}
		\sqrt{ \varepsilon_{23}} = \sqrt{2q_D\cdot\left(4 \big((q_H+q_D) + (q_D+1) q_H \big)\sqrt{2 g_{sk}}+2^{-n}\right)} 
		\leq&6\sqrt{q_Hq_D}\cdot \left(g_{sk}^{1/4}+2^{-n/2}\right)\\
		\leq& 12\sqrt{qq_D}\cdot g_{sk}^{1/4} \, ,
		\end{align*}
	where we have used the fact that $2^{-n}\le g_{sk}\le 1$ in the last line.
	Taking the expectation over $(sk,pk)\leftarrow \mathsf{Gen}$, applying Jensen's inequality and using $q_H+q_D\le q$ once more,
		we get the claimed bound. 
		Finally, we note that the runtime of $\cal B$ is given by $T_{\cal B} = T_{\cal A} + T_{\textsc{Decaps}} + T_{G} + T_{\cal S}$, where apart from its oracle queries {\sc Decaps} runs in time linear in $q_D$, and $\cal S$ can be simulated in time  
		$$
		T_{\mathcal S}= O\bigl(q_{RO} \cdot q_E\cdot \mathrm{Time}[f] + q_{RO}^2\bigr) = O\bigl(q_H\cdot q_D\cdot \mathrm{Time}[{\sf Enc}] + q^2\bigr)
		$$
		by Theorem \ref{thm:MainFeatures}, and similarly for $G$. 
		\qed
	\end{proof}

\endgroup

\submission{
\section{Acknowledgement}
The authors thank Andreas Hülsing and Kathrin Hövelmanns for helpful
discussions, and Eike Kiltz and anonymous referees for helpful comments on an earlier version of this
article.
JD was funded by ERC-ADG project 740972 (ALGSTRONGCRYPTO). 
SF was partly supported by the EU Horizon 2020 Research and Innovation	Program Grant 780701 (PROMETHEUS).
CM was funded by a NWO VENI grant (Project No. VI.Veni.192.159). 
CS was supported by a NWO VIDI grant (Project No. 639.022.519).
}{}

\switch{\bibliographystyle{alpha}}{\bibliographystyle{abbrv}}
\bibliography{QROM}

\appendix

\section*{\switch{Appendix}{SUPPLEMENTARY MATERIAL}}

\section{A gap in the security proof from \cite{Zhandry2018} for the FO transformation} \label{app:gap}
In his seminal paper \cite{Zhandry2018}, Zhandry introduced the so-called compressed-oracle technique, a ground-breaking method that led to many new results in post-quantum cryptography, quantum query complexity and beyond.  One of the most important features of the compressed-oracle methodology is that it allows the approximate recovery of  several features of the classical ROM, that were previously believed lost when moving to the QROM.

The new, ``virtually classical'' ways of reasoning about quantum access to a random oracle are very intuitive. This fact bears a certain risk that the reach of classical intuition in the compressed-oracle framework is overestimated. In the following, we describe a gap in the security proof for the Fujisaki-Okamoto (FO) transformation given in \cite{Zhandry2018}, which was likely caused by following the classical intuition too closely.

One step in security reductions for the FO transformation is the simulation of the decryption or decapsulation oracle without making use of the secret key. This simulation is done by accessing (either actively by programming, or passively by preimage awareness) the adversary's random-oracle interface. For proofs in the QROM, the adversary's queries cannot be compiled into a list in a straight-forward manner (due to the no-cloning principle, if you will). If a reduction collects information about an adversary's QROM queries \emph{during runtime}, be it by directly accessing the adversary's query input or output, or by acting on the compressed-oracle register, it needs to be analyzed to which degree the information-collection operation can be noticed by the adversary. 

In the security proof for the FO transformation in \cite{Zhandry2018}, the replacement of the decryption oracle by a simulated version happens gradually in Hybrids 2 to 4 (Lemma 43 and 44 in the full version of~\cite{Zhandry2018}). 
In more detail, in Hybrid 2 a (purified) ``test'' is performed on the state of the compressed oracle before the reply to the decryption query is prepared and sent, and then uncomputed again right afterwards; since (due to Lemma~39 of~\cite{Zhandry2018}) the uncomputation almost commutes with the re-encryption check performed as part of the preparation of the reply, this ``test'' and its uncomputation have negligibe effect. 
In Hybrid 3, the result of the ``test'' is then used in the derivation of the reply to the decryption query by setting the reply to $\bot$ in case the ``test'' fails. Finally, in Hybrid 4, it is declared that the (simulated) decryption oracle  
\emph{``scans over the inputs of the [compressed oracle] database for $G$,
 looking for inputs [of a certain form]. For each one, we will check if [it encrypts to the queried ciphertext]''}; the first database entry where the check succeeds is then used to answer the query. 
 
Using a more formal language, in each of these hybrids the reply to the decryption query is obtained by means of applying a measurement to the state of the compressed oracle (where the measurement depends on the queried ciphertext $c$, and on the secret key in Hybrids 2 and 3). In Hybrid 2, the measurement consists of the ``test'', the (ordinary) derivation of the oracle response, and the uncomputation of the ``test''. At the other end, in Hybrid 4, it consists of all the ``scanning'' and ``checking'' etc. 
By the nature of quantum measurements, 
in both steps, from Hybrid 2 to 3 and from Hybrid 3 to 4, {\em both} the reply of the (simulated) decryption oracle {\em and} the post-measurement state of the compressed oracle (and thus the future behavior of the compressed oracle) may change. While in the proof in~\cite{Zhandry2018} it is argued for both steps, from Hybrid 2 to 3 and from Hybrid 3 to 4, that the reply of the (simulated) decryption oracle does (almost) not change, for neither of the two steps is it argued that the post-measurement state is not (much) affected. As a matter of fact, Hybrids 3 and 4 are described in such a ``virtually classical'' way that there is ambiguity to translate them into proper descriptions of quantum measurements, necessary to analyze the effect on the post-measurement state.  
 

It seems to us that completing the proof in~\cite{Zhandry2018}, which requires to rigorously specify the respective quantum measurements in Hybrids 3 and 4 
and to analyze the resulting disturbance of the state of the compressed oracle, is non-trivial. 
Given the informal description of the hybrids, we find it hard to judge whether it is ``only'' a question of filling in the gaps, or whether the claimed indistinguishability of the hybrids is actually false (our proof uses a different sequence of hybrids). 

Exactly the same problem exists in recent follow-up work by Katsumata, Kwiatkowski, Pintore and Prest~\cite{KKPP20}, who follow the FO proof outline from~\cite{Zhandry2018}.

\section{Efficient representation of the compressed oracle. }\label{subsec:compressed}

\def\Enc{\mathsf{SparseEnc}}

By the techniques of \cite{Zhandry2018}, it is possible to make the (considered variant of the) compressed oracle efficient. Concretely, by means of a suitable encoding, it is possible to {\em efficiently} maintain the quantum state of the register $D$ of the compressed oracle, compute the unitary $O_{XYD}$, and extract information from the state of $D$. We briefly describe this procedure below. 

Writing $\bar{\cal Y} = \{0,1\}^n \cup \{\bot\}$, consider the following standard  sparse encoding scheme 
$$
\Enc^q: \bar{\cal Y}^{{\cal X}} \to {\cal D} = ({\cal X} \times \bar{\cal Y})^q \, ,
$$
which maps any ``database'' ${\bf y}  = (y_x)_{x \in \cal X}$ with at most $q$ non-$\bot$ entries to the ``compressed database''
$$
\Enc^q({\bf y}) = \big((x_1,y_{x_1}),\ldots,(x_s,y_{x_s}),(0,\bot),\ldots,(0,\bot)\big)
$$ 
of pairs $(x,y_{x})$ with $y_{x} \neq \bot$, sorted as $x_1 < \cdots < x_s$, and padded with $(0,\bot)$s. 
Naturally, we then set 
$$
\Ket{\Enc^q({\bf y})} = \ket{x_1}\ket{y_{x_1}} \cdots \ket{x_s}\ket{y_{x_s}} \ket{0}\ket{\bot} \cdots \ket{0}\ket{\bot} \in \big(\C[{\cal X}] \otimes \C[\bar{\cal Y}]\big)^{\otimes q}
$$
for any such $\bf y$. The crucial observations now are: 

\begin{enumerate}\setlength{\parskip}{1ex}
\item Using the representation $H^{\otimes |{\cal X}|} \ket{{\bf y}} \mapsto \ket{\Enc^q({\bf y})}$ for the state of register $D$ after $q$ queries, the evolution of the compressed oracle, given by $O_{XYD}$, is an efficiently quantum computable isometry (this was shown by Zhandry, but is also easy to see from scratch). Here and below, $H$ is the Walsh-Hadamard transform on $\C[\{0,1\}^n] = (\C^2)^{\otimes n}$, extended to act as identity on $\ket\bot$. 
\item Using the representation $\ket{{\bf y}} \mapsto \ket{\Enc^q({\bf y})}$ instead, it follows from basic theory of quantum computation that for any classical function $f$ with domain $\bar{\cal Y}^{{\cal X}}$ and that is classically efficiently computable using the representation ${\bf y} \mapsto \Enc^q({\bf y})$, the unitary $U:\ket{{\bf y}}\ket{z} \mapsto \ket{{\bf y}}\ket{z + f({\bf y})}$ is efficiently quantum computable. 
\item $\ket{{\bf y}} \mapsto \ket{\Enc^q({\bf y})}$ commutes with applying Walsh-Hadamards to the $\C[\bar{\cal Y}]$-components. Therefore, one can efficiently switch between the two representations above, simply by applying $H^{\otimes q}$ to the corresponding registers of $\ket{\Enc^q({\bf y})}$. 
\end{enumerate} 
Thus, using either of the two representations for representing the internal state of the oracle, both the evolution of the oracle and the typical unitaries or measurements used to ``read out'' information are efficiently quantum computable. For example, checking if $y_x = \bot$ for a given $x \in \cal X$, or if there exists $x \in \cal X$ for which $x$ and $y_x$ satisfy some given (efficiently computable) relation, etc. 
Formally:

\begin{lemma}\label{lem:efficient-sparse}
	Let $f:\left(\{0,1\}^n\cup\{\bot\}\right)^{|\mathcal X|}\to\mathcal T$ be a function such that $\tilde f = f\circ\mathsf{SparseDec}^q$ can be computed in polynomial time in $q$. Then the measurement $\{\tilde \Pi^t\}_{t \in \cal T}$ given by the projections 
	$$
	\tilde \Pi^t = \sum_{{\bf y}:\tilde f({\bf y})=t} \proj{{\bf y}}
	$$
	can be implemented in time  linear in $\mathrm{Time}[\tilde f]$ and thus in quantum polynomial  time in  $q$. 
\end{lemma}

\section{Supplementary proofs}\label{sec:SuppProofs}

\subsection{Proof of Lemma \ref{lem:simple}}

	Recalling from (\ref{eq:F}) that $F\ket y=\ket y+2^{-n/2}\ket\delta$ with $\ket\delta := \ket\bot - \ket{\phi_0}$, we have
	$$
	[F,\proj y] = F \proj y - \proj y F 
	= 2^{-n/2}\ketbra{\delta}{y} - 2^{-n/2} \ketbra{y}{\delta} \, .
	$$
	From this, it follows that
	$$
	[F, \Pi^{x} ] = \!\sum_{y\in\{0,1\}^n\atop (x,y)\in R}\!\!\! [F,\proj y]  \leq 2^{-n/2} \, \ket\delta \!\!\!\!\sum_{y\in\{0,1\}^n\atop (x,y)\in R}\!\!\!  \bra y - 2^{-n/2} \!\!\!\!\sum_{y\in\{0,1\}^n\atop (x,y)\in R}\!\!\!  \ket y \bra\delta
	$$
	and thus, using (\ref{eq:norminequality}), that
	$$
	\|[F, \Pi^{x} ]\| \leq 2^{-n/2} \, \|\ket\delta\| \bigg\| \sum_{y\in\{0,1\}^n\atop (x,y)\in R} \!\!\! \bra y \bigg\| \leq 2^{-n/2}\sqrt{2}\sqrt{\Gamma_x} \, .
	$$
	
	
	For the second bound, let  $C_{Y D_x} = \cnot$ with $\cnot$ as in (\ref{eq:defO}), with the understanding that $D_x$ is the control register and $Y$ the target. Recall from \eqref{eq:defO} that $O^x_{Y D_x} = F_{D_x} C_{Y D_x} F_{D_x}$. Thus, using (\ref{eq:CommutatorOfProduct}) twice and omitting the registers, we obtain 
	$$
	[O^x,\Pi^x] = F[CF,\Pi^x] + [F,\Pi^x] CF = FC[F,\Pi^x] + F[C,\Pi^x]F + [F,\Pi^x] CF \, .
	$$
	Finally, we notice that $[C_{Y D_x},\Pi^x_{D_x}] = 0$, since projections on the control register of a CNOT commute with the CNOT. The claimed bound now follows from the derived bound on $[F,\Pi^x]$ together with Equation \eqref{eq:ComOfTensorProduct}.
	
	
	The third bound follows by recalling that \smash{$\Pi^{\noinstancesuperscript}_D = \bigotimes_{x'} \bar\Pi_{D_{x'}}^{x'}$} is a tensor-product for which $O^x_{YD_x}$ acts trivially on all the components except for the component $\bar\Pi_{D_x}^x$, so with Equation \eqref{eq:ComOfTensorProduct} we obtain, 
	$$
	\| [O^x_{YD_x},\Pi^{\noinstancesuperscript}_D] \| \leq  \| [O^x_{YD_x}, \bar\Pi^{x}_{D_x}] \| =  \| [O^x_{YD_x}, \Pi^{x}_{D_x}] \| \, ,
	$$
	%
	%
	which completes the proof. \qed

\subsection{Proof of Proposition~\ref{prop:HardCollision}}

The left circuit in Fig.~\ref{fig:CircuitsHardCollision} defines (the distribution of) the considered variables $x, \hat x, h, t$. We also consider the circuit that applies the measurement $\{\Pi^{col}, \Pi^{\neg col}\}$ instead of ${\cal M}^t$, where $\Pi^{col}$ is as in Lemma~\ref{lem:ColBound} and $\Pi^{\neg col}\ = \id -\Pi^{col}$ (Fig.~\ref{fig:CircuitsHardCollision}, middle). Since the projections defining either measurement are all diagonal in the basis $\{\ket{\bf y}\}$, we may equivalently measure register $D$ in that basis (Fig.~\ref{fig:CircuitsHardCollision}, right), and then set $\hat x$ to be the smallest element $\cal X$ so that $f(\hat x,y_{\hat x}) = t$ (with $\hat x = \emptyset$ if no such element exists) and consider the event $col$ given by $\exists \, x' \neq x'' : f(x',y_{x'}) = f(x'',y_{x''})$. 
By the respective definitions of ${\cal M}^t$ and $\Pi^{col}$, both, the variables $\hat x, x, h, t$ and the event and variable $col$ and $x,h,t$ then have the same distributions as in the respective original two games. But now, we can consider their joint distribution and argue that 
\begin{align*}
\Pr[\hat x \neq x \wedge f(x,h) = t ] 
		\leq \Pr[\hat x \neq x  \,|\, f(x,h)=t \wedge \neg col ]+ \Pr[col]  \, .
\end{align*} 
We now observe that right before the considered measurement, by definition of $O$, the state of $D$ is supported by vectors $F\ket{\bf y}$ with $y_x = h$ (here we use the assumption that no previous extraction queries have been made, see Preliminaries for further detail), and so the measurement outcome $\bf y$ satisfies $y_x = h$ with probability
$1 - 2\cdot 2^{-n}$ by Equation (\ref{Eq:MeasureClassicalQuery}).
Therefore, the first term is bounded by $2\cdot 2^{-n}$ by definition of $col$ and $\hat x$, while $\Pr[col]$ is bounded by $\frac{40 e^2 (q+2)^3 \Gamma'(f)+2}{2^n}$, using Lemma~\ref{lem:ColBound}.
\qed

\begin{figure}[h]		
	$$
	\Qcircuit @C=0.4em @R=.25em @!R {
		&\lstick{D}  & \qw                            &  \multigate{2}{O}     &  \qw & \push{\!...\!\!}  &  &    \qw                    & \qw        &  \qw           & \qw   & \multigate{2}{O}      & \qw       &  \measuretab{\!{\cal M}^t\!}  &  \rstick{\hat x}\cw  \\
		&\lstick{X}  &  \multigate{2}{\!A_0\!}    &  \ghost{O}               & \qw & \push{\!...\!\!}   & & \multigate{2}{\!A_q\!}  & \cw       &    \push{\!x\!}  &    &   \cghost{O}         &  \rstick{x}\cw   \\
		&\lstick{Y}  &  \ghost{\!A_0\!}               &  \ghost{O}               & \qw  & \push{\!...\!\!}  & &  \ghost{\!A_q\!}           &               & \push{\!0\!}   &   & \cghost{O}         &  \rstick{h}\cw   \\
		&                &   \ghost{\!A_0\!}             & \qw                        & \qw  & \push{\!...\!\!}  &  &   \ghost{\!A_q\!}           &  \cw      &  \push{\!t\!}   &  &   \cw                  & \cw    &  \cctrl{-3}              &  \rstick{t}\cw     \\
	}
	\qquad \qquad
	\Qcircuit @C=0.4em @R=.25em @!R {
		& \lstick{...\!\!}   & \multigate{2}{O}      & \qw           & \measuretab{\!\Pi^{col}\!}  & \cw  \\
		& \lstick{...\!\!}   &   \cghost{O}         &  \rstick{x}\cw   \\
		& \lstick{...\!\!}   & \cghost{O}         &  \rstick{h}\cw   \\
		& \lstick{...\!\!}  &  \cw  &      \cw   & \cw     &  \rstick{t}\cw     \\
	}
	\qquad \qquad
	\Qcircuit @C=0.4em @R=.25em @!R {
		& \lstick{...\!\!}   & \multigate{2}{O}      & \qw      & \meter  &  \rstick{{\bf y}  \leadsto \hat x}\cw  \\
		& \lstick{...\!\!}     &   \cghost{O}         &  \rstick{x}\cw   \\
		& \lstick{...\!\!}   & \cghost{O}         &  \rstick{h}\cw   \\
		& \lstick{...\!\!}   &       \cw  &       \cw  & \lstick{t}         \\
	}\qquad\quad
	\vspace{-2ex}
	$$
	\caption{Quantum circuit diagrams for the experiments in the proof of Prop.~\ref{prop:HardCollision}. }\label{fig:CircuitsHardCollision}
\end{figure}
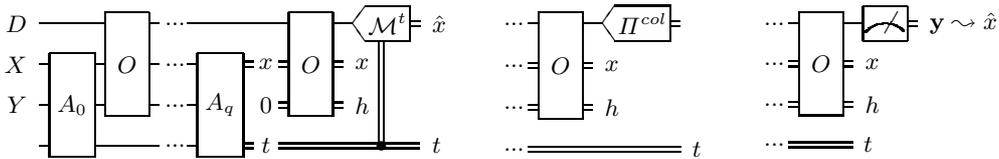

\section{Hardness of collision finding}

The following can be easily extracted from the derivation of the general collision-finding bound Theorem~5.29 from \cite{CFHL20}. It expresses that, for any algorithm with bounded query complexity, it is unlikely that one encounters a collision within the superposition oracle.  
\begin{lemma}\label{lem:ColBound}
Let $f:\mathcal X \times \mathcal \{0,1\}^n \to \mathcal T$, and let $\Pi^{col}$ be the projection into the space spanned by $\ket{{\bf y}} \in \H_D$ for ${\bf y} = (y_x)_{x \in \cal X} \in ({\cal Y} \cup \{\bot\})^{\cal X}$ such that there exist $x \neq x'$ with $y_x,y_{x'} \neq \bot$ and $f(x,y_x) = f(x',y_{x'})$. Then, for any oracle algorithm $\cal A$ with query complexity $q$, at the end of the execution the state $\rho$ of the compressed oracle is such that
$$
\tr(\Pi^{col} \rho) \leq 40 e^2 q^2 (q+1) \Gamma'(f) / 2^{n} \, ,
$$ 
where $\Gamma'(f) = \! \displaystyle\max_{x \neq x' , y'} | \{y \mid f(x,y)= f(x',y') \} |$ and $e \approx 2.718$ is Euler's number.  
\end{lemma}

	\section{Application to Fiat Shamir Signatures}\label{Appendix:ApplSig}
	
	\sigps are commonly used to obtain non-interactive zero-knowledge proofs and digital signatures via the Fiat Shamir (FS) transform. Here, the random challenges are (possibly after a suitable number of parallel repetitions) replaced by the hash of the first message in the 3-round protocol, thus making the protocol non-interactive. To construct a digital signature scheme (DSS), the message to be signed is included in the hash argument.\footnote{For FS DSS, the relation $R$ needs to admit an efficient generator of hard instances.}
	
	The post-quantum security of FS signatures has recently drawn additional attention
	. This is mainly because FS signatures are some of the most promising candidates for replacing RSA and elliptic curve signatures which can be broken by quantum adversaries. Indeed, two out of the 6 round-3 candidate DSSs in the NIST standardization process for post-quantum cryptographic schemes, CRYSTALS Dilithium \cite{DLLSSS18} and Picnic \cite{Chase2017}, are FS signature schemes. In the QROM,\footnote{The typical ROM reductions proceed similarly} the chain of arguments for reducing the UF-CMA security of a FS signature scheme $\mathsf{Sig}[\Sigma]$ to%
\submission{
the i) honest-verifier zero-knowledge, and ii) (some variant of the) special soundness, properties of the underlying \sigp $\Sigma$ as follows
	(also depicted in Fig.~\ref{fig:FS-sigs}).
	\begin{itemize}\setlength{\parskip}{0.5ex}
		\item First, the UF-CMA security of $\mathsf{Sig}[\Sigma]$ is reduced to plain unforgeability (UF-NMA), using the HVZK property of  $\Sigma$  \cite{Kiltz2017,GHHM20}.
		\item The UF-NMA property of  $\mathsf{Sig}[\Sigma]$ follows from the extractability of the Fiat Shamir transformation $\mathsf{FS}[\Sigma]$ of  $\Sigma$.
		\item The extractability of $\mathsf{FS}[\Sigma]$ is then reduced to the extractability of $\Sigma$ \cite{DFMS19,LZ19a,DFM20}.
		\item Finally, the extractability of $\Sigma$ is reduced to the (variant of) special soundness of $\Sigma$ \cite{Unruh2012}.
	\end{itemize} 
	}
	{
	(some variant of) the special soundness property of the underlying \sigp $\Sigma$ is as follows (where the first step, from UF-NMA to UF-CMA security, additionally requires $\Sigma$ to be honest-verifier zero-knowledge):
	}
	
\begin{figure}
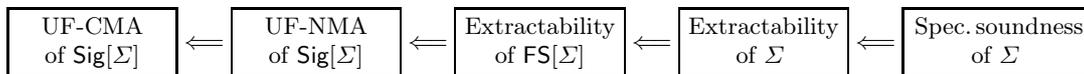

\vspace{-2ex}
$$
\submission{
\fbox{\begin{minipage}{2cm}
\small\center
UF-CMA \\ of $\mathsf{Sig}[\Sigma]$
\end{minipage}}
\Longleftarrow
\fbox{\begin{minipage}{2cm}
\small\center
UF-NMA \\ of $\mathsf{Sig}[\Sigma]$
\end{minipage}}
\Longleftarrow
\fbox{\begin{minipage}{2cm}
\small\center
Extractability of $\mathsf{FS}[\Sigma]$
\end{minipage}}
\Longleftarrow
\fbox{\begin{minipage}{2cm}
\small\center
Extractability of $\Sigma\phantom{]}$ 
\end{minipage}}
\Longleftarrow
\fbox{\begin{minipage}{2.3cm}
\small\center
Spec.\,soundness of $\Sigma\phantom{]}$
\end{minipage}}	
}{
\fbox{\begin{minipage}{1.4cm}
\small\center
UF-CMA \\ of $\mathsf{Sig}[\Sigma]$
\end{minipage}}
\Longleftarrow
\fbox{\begin{minipage}{1.4cm}
\small\center
UF-NMA \\ of $\mathsf{Sig}[\Sigma]$
\end{minipage}}
\Longleftarrow
\fbox{\begin{minipage}{2cm}
\small\center
Extractability of $\mathsf{FS}[\Sigma]$
\end{minipage}}
\Longleftarrow
\fbox{\begin{minipage}{2cm}
\small\center
Extractability of $\Sigma\phantom{]}$ 
\end{minipage}}
\Longleftarrow
\fbox{\begin{minipage}{2.3cm}
\small\center
Spec.\,soundness of $\Sigma\phantom{]}$
\end{minipage}}	
}
$$
\caption{Chain of arguments for proving security of FS signatures. }\label{fig:FS-sigs}
\end{figure}

Prior to this work, the last step (arguing extractability from special soundness) has relied on Unruh's rewinding lemma~\cite{Unruh2012}, which after suitable generalization leads, e.g., to a $2k+1$-th root loss for a $k$-sound $\Sigma$. For commit-and-open \sigps, Theorem \ref{thm:CnO-extract} can replace Unruhs rewinding lemma when working in the QROM, making the last step above tight up to unavoidable additive errors.

As an example, Theorem \ref{thm:CnO-extract} implies a sizeable improvement over the current best QROM security proof of Picnic2 \cite{Chase2017,KZ20,CD+20}. Indeed, Unruh's rewinding lemma implies a 6-th root loss for the variant of special soundness the underlying \sigp possesses \cite{DFMS19}, while Theorem~\ref{thm:CnO-extract} is tight. 

We note that for commit-and-open \sigps, there is hope for further improvements by means of combining the last two steps and doing a {\em direct} analysis of $\mathsf{FS}[\Sigma]$. Indeed, \cite{Chailloux20} suggests such an approach, but the analysis provided there there still relies on some unproven assumption.  
	
\section{Public-Key Encryption and Key Encapsulation}\label{AppendixFO}
Following the presentation of \cite{HHK17} in general lines, we recall the formal definition of a public-key encryption scheme.

\begin{definition}[Public-Key Encryption]
	A  \emph{public-key encryption scheme} $\sf PKE$ consists of algorithms $\sf (Gen,Enc,Dec)$, a message space ${\cal M}$, a ciphertext space ${\cal C}$ and a set of random coins $\cal R$, such that for any $m\in \cal M$, $r\in {\cal R}$
	\begin{align*}
	(sk,pk) \leftarrow {\sf Gen} \; , \quad 	{\cal C}\ni c \leftarrow {\sf {Enc}}_{pk}(m;r) \quad\text{and}\quad
	{\sf Dec}_{sk}(c) \in {\cal M \cup \{\bot\}} \, .
	\end{align*}
\end{definition}

For a given public-key encryption scheme, it may be useful to consider the probability of encountering decryption failures. 

\begin{definition}[$\delta$-correctness] A public-key encryption scheme is $\delta$-correct if
	$$
	\E_{(sk,pk)\leftarrow \sf Gen}\biggl[\max_{m\in \cal M}\Pr\bigl[{\sf Dec}_{sk}(c) \neq m : c\leftarrow {\sf Enc}_{pk}(m)\bigr] \biggr]\leq \delta
	$$
	where the probability is over the randomness of the encryption.
\end{definition}

Another important property of encryption schemes is the min-entropy of a ciphertext given the plaintext, measured by their {\em $\gamma$-spreadness}. 

\begin{definition}[$\gamma$-spreadness] A public-key encryption scheme is $\gamma$-spread if
	$$
	\min_{m\in {\cal M} \atop (sk,pk) } \Bigl(-\log \max_{c\in \cal C}\Pr\bigl[ c={\sf Enc}_{pk}(m)\bigr] \Bigr) \geq \gamma \, ,
	$$
	where the probability is over the randomness of the encryption, and the minimum is over all key pairs that have positive probability of being produced by $\sf Gen$.
\end{definition}

The above definition can be relaxed to an {\em expectation} over the choice of $pk$, when the expectation is done inside the negative logarithm. 
\begin{definition}[weak $\gamma$-spreadness]\label{def:weak-gamma} A public-key encryption scheme is weakly $\gamma$-spread if
	$$
	-\log \E_{(sk,pk)\leftarrow \sf Gen}\biggl[ \max_{m\in {\cal M} \atop c\in \cal C}  \Pr\bigl[ c={\sf Enc}_{pk}(m)\bigr] \biggr] \geq \gamma \, ,
	$$
where again the probability is over the randomness of the encryption
.
\end{definition}

A key-encapsulation mechanism (KEM) is defined as follows:

\begin{definition}[Key Encapsulation Mechanism]\sloppy
	A \emph{key encapsulation mechanism} $\sf KEM$ consists of algorithms $(\sf Gen,Encaps,Decaps)$ and a key space ${\cal K}$, where
	\begin{align*}
	(sk,pk) \leftarrow {\sf Gen} \; , \quad 
	(K,c) \leftarrow {\sf Encaps}(pk) \quad\text{and}\quad
	{\sf Decaps}_{sk}(c) \in {\cal K \cup \{\bot\}} \, .
	\end{align*}
\end{definition}

\end{document}